\documentclass[12pt]{article}
\usepackage{mydef}
\usepackage{mathtools}
\usepackage{booktabs}
\usepackage{wrapfig}
\usepackage{lscape}
\usepackage{rotating}
\usepackage[hide]{ed}
\usepackage{appendix}

\usepackage{tikz}
\usetikzlibrary{calc,decorations.markings}
\usetikzlibrary{arrows,positioning}

\newcommand{\ra}[1]{\renewcommand{\arraystretch}{#1}}

\begin{document}
\title{Efficient solution of structural default models with
correlated jumps and mutual obligations}

\author[1,2]{Andrey Itkin}
\author[1,3]{Alex Lipton}
\affil[1]{\small Bank of America Merrill Lynch}
\affil[2]{\small New York University, School of Engineering}
\affil[3]{\small Oxford-Man Institute of Quantitative Finance, University of Oxford}

\date{\today}

\maketitle

\blfootnote{The views represented herein are the authors' own views and do
not necessarily represent the views of BAML or its affiliates and are not a product of BAML Research.}

\begin{abstract}
The structural default model of \cite{LiptonSepp2009a} is generalized for a
set of banks with mutual interbank liabilities whose assets are driven by
correlated L\'{e}vy processes with idiosyncratic and common components. The
multi-dimensional problem is made tractable via a novel computational
method, which generalizes the one-dimensional fractional partial
differential equation method of \cite{Itkin2014} to the two- and
three-dimensional cases. This method is unconditionally stable and of the second order of
approximation in space and time; in addition, for many popular L\'{e}vy models it has linear complexity in each dimension. Marginal and joint survival probabilities for two and three banks with
mutual liabilities are computed. The effects of mutual liabilities are
discussed, and numerical examples are given to illustrate these effects.
\end{abstract}

\section{Introduction}
Structural default framework is widely used for assessing credit risk of
rate debt. Introduced in its simplest form in a seminal work \cite{m74},
this framework was further extended in various papers, see a survey in
\cite{LiptonSepp2011} and references therein. In contrast to reduced-form
models, structural default models suffer from the curse of dimensionality
when the number of counterparties grows; however, these models provide a
more natural financial description of the default event for a typical firm.

One of the possible extensions of the structural framework, which is of high
importance in the current environment, consists in taking into account the
fact that banks, in addition to their liabilities to the outside economy,
also have some liabilities to each other. This topic is discussed, e.g., in
\cite{BOE2011}, where it is mentioned that systemic capital requirements for
individual banks, determined as the solution to the policymaker's
optimization problem, depend on the structure of banks'\ balance sheets
(including their obligations to other banks) and the extent to which their
asset values tend to move together. More generally, systemic capital
requirements are found to be increasing in banks' balance sheet size
relative to other banks in the system, as well as their interconnectedness,
and, materially, contagious bankruptcy costs.

From this perspective, an extension of the simplest Merton model can be
proposed to quantify default risks in an interconnected banking system. For
instance, \cite{Elsinger2006, Gauthier2010} consider systemic risk in such a
system and attribute it either to correlations between the asset values of
the banks, or to interlinkages of the banks' balance sheets, which could
result in contagious defaults. An extended Merton model can be built as a
combination of the correlated Merton balance sheet models, calibrated by
using observed bank equity returns, and a network of interbank exposures
cleared in the spirit of \cite{EisenbergNoe2001}.

In this paper we develop a model, which builds upon its predecessors; yet,
it differs from the earlier models in one very important respect. Namely,
rather than addressing a point-in-time default event, we consider defaults,
which can occur at any time, by introducing a continuous default barrier in
the spirit of \cite{BC1976}. We feel that this extension is necessary in
order to analyze the effect of mutual liabilities properly, especially
because we wish to provide not just qualitative, but also quantitative
conclusions. To avoid confusion, we emphasize that this effect differs from
that of contagion for correlated defaults in reduced-form models, see,
e.g., \cite{Yu2007, Bielecki2011}.

To achieve our goal, we need to come up with a suitable structural model
capable of handling mutual obligation effects at various time scales. It is
well-known that pure diffusion asset dynamics is manifestly inadequate for
relatively short time-scales, and we need to introduce jumps into the
model, see, e.g., \cite{Zhou2001, Lipton2002}. Therefore, we choose a L\'{e}vy
jump-diffusion driver for the asset dynamics.

Multi-dimensional L\'{e}vy processes find various applications in
mathematical finance. They are used in modeling basket equity derivatives,
various credit derivatives, etc. Unfortunately, tractability of
multi-dimensional L\'{e}vy processes is rather limited. In addition, it is
difficult to study such processes because they suffer from the curse of
dimensionality. Various numerical methods, including analytical,
semi-analytical, finite-difference (FD), Monte-Carlo methods, and their
combinations have been used for solving the corresponding problems, again
see, e.g., a survey in \cite{LiptonSepp2011} and references therein.
Certainly, rather straightforward Monte Carlo method can be proposed to
simulate multi-dimensional L\'{e}vy processes. However, in general it is
both slow or inaccurate. Therefore, finite difference methods seem to be a
viable alternative for 2D and 3D problems, despite the fact that in the 3D
case such methods can be relatively slow (but definitely faster than the
corresponding Monte Carlo method).

The authors are aware of limited number of papers on mathematical finance, which are
using FD methods to solve 2D partial Integro-differential equation (PIDE) describing
the evolution of two fully correlated assets, see, e.g., \cite{CliftForsyth2008}, \cite{LiptonSepp2009a, LiptonSepp2011}. In \cite{CliftForsyth2008}, the authors
use a bivariate distribution proposed in \cite{MarshallOlkin1967} and
consider normal and exponentially distributed multivariate jumps. In \cite{LiptonSepp2009a},
the authors consider assets, which are correlated
twofold. First, diffusion components are correlated in the standard manner
because they are driven by correlated Brownian motions. Second, jump
components are correlated because for each asset they are represented as a
sum of a) systemic exponential jumps common for all assets, and b)
idiosyncratic exponential jumps specific for a particular asset. From a
historical perspective, this idea can be traced back to the work of Vasicek,
who developed a multi-factor structural model assuming that the dynamics of
individual assets can be described as a sum of systemic and idiosyncratic
parts, \cite{vasicek1987,vasicek2002}
\footnote{It should be emphasized that Vasicek model considers a single period setting,
whereas L\'{e}vy models have to be analyzed in continuous time. In addition,
L\'{e}vy models use infinitely divisible distributions, rather than standard Gaussian random variables.}.

However, other L\'{e}vy models could be of interest as well, see, e.g., \cite{EberleinKeller:95}, where it is shown that generalized hyperbolic models fit the market data pretty well. Therefore, an extended framework which allows for general L\'{e}vy models to be used when modeling jumps is highly desirable. Below we provide a short survey of various approaches to introduce multivariate correlated jumps via L\'{e}vy's copula, multivariate subordinators of the Brownian motion, etc., as well as discuss their advantages and pitfalls. Our main concerns with regard to the existing approaches are two-fold: a) some of them are not flexible enough to meet all the modeling requirements, because they impose some undesirable restrictions on the jump correlation structure; b) they suffer from the curse of dimensionality in the sense that their complexity is polynomial rather than linear in each dimension.

Another observation is that even in the 1D case traditional methods for
solving PIDEs experience some problems, see a survey in \cite{Itkin2014},
and references therein. In the multi-dimensional case these problems become
even harder. To deal with these problems, we choose a particular way of
introducing correlated jumps and combine it with the multi-dimensional
version of the matrix exponential method proposed first in \cite{ItkinCarr2012Kinky}
and later further elaborated in \cite{Itkin2014,Itkin2014a}. The presented construction allows different jumps to be used for modeling the idiosyncratic and common factors. For example, in
the 2D case we can represent idiosyncratic jumps of the first bank by using
the Meixner model of \cite{Schoutens01}, idiosyncratic jumps of the second
bank by using the Merton model, and simulate their common jumps by using the
CGMY model. We do not claim that such rich choice of L\'{e}vy processes is
necessary in practice, since the actual jump distribution is hard to
establish with certainty, merely that it is possible to do. In our experience,
hyper-exponential jumps introduced in \cite{Lipton2002} are more than
adequate for all practical purposes. We don't consider every
possible combination of L\'{e}vy processes in this paper, since this could
be done based on the general principles described in \cite{Itkin2014,Itkin2014a}.
However, as an example, we do consider a model with
Gaussian idiosyncratic jumps and exponential systemic jumps. As part of this example, we think
of idiosyncratic jumps as two-sided, while systemic jumps as one-sided. In
this sense, our example should be ideologically similar to that in \cite{CliftForsyth2008}.
However, our method is not restricted by this choice and
differs from that of \cite{CliftForsyth2008} in several important respects:
a) we use Gaussian and exponential jumps just as an example, other common
jumps and univariate marginals could be used as well; b) we use the matrix
exponential method, rather than the traditional method for solving the
corresponding PIDE; c) we present a splitting method to provide solutions
of the 2D and 3D problems with second order of accuracy in both space and
time, and prove convergence of the method. Our method is of the linear
complexity (i.e., $O(N_{1}\times N_{2})$ in the 2D case and $O(N_{1}\times
N_{2}\times N_{3})$ in the 3D case) provided that the Merton, Kou, CGMY or
Meixner L\'{e}vy models are used. Our method is faster than the FFT method
used in \cite{CliftForsyth2008}.

In this paper, we concentrate on our structural default model for two or
three banks with mutual liabilities. The method can also be used to price
basket options. We show that accounting for these liabilities affects both
the joint survival probability of these banks, which is to be expected, as
well as their marginal survival probabilities, which is not the case when
mutual liabilities are ignored. This fact has to be taken into account when
marginals are calibrated to the market CDS spreads. We provide several
numerical examples in order to demonstrate that the presence of mutual
obligations could potentially strongly affect the corresponding survival
probabilities, and, by implication, the stability of the inter-bank system,
especially in the 3D case.

The new results of the paper are as follows: a) interbank mutual
obligations are incorporated in the structural default credit model with
correlated jumps, and their impact on the joint and marginal probabilities
is investigated both qualitatively and quantitatively; b) new splitting
method is proposed to solve the corresponding PIDE with correlated jumps in
the 2D and 3D cases. The method includes new steps that don't appear in the
1D case. For many popular L\'{e}vy models the method provides linear
complexity in each dimension and is unconditionally stable.

The rest of the paper is organized as follows. In section~\ref{stat} we
describe our multi-dimensional structural model, which is an extension of
\cite{LiptonSepp2009a}. In section \ref{secCorJumps} we provide a short
survey of the existing approaches to multivariate correlated jumps, and
describe the one we find to be particularly suitable for our goals. In
section~\ref{Frac} we shortly describe the method of \cite{ItkinCarr2012Kinky, Itkin2014,Itkin2014a} and extend it to the multi-dimensional case. In section~\ref{splitting} we describe the splitting
algorithm, which is adopted for solving the corresponding multi-dimensional
PIDE. In section~\ref{NP} we provide a detailed numerical scheme for solving
the fractional jump equations and prove the unconditional stability, second
order accuracy and convergence of the scheme. We also emphasize that our
scheme preserves positivity of the solution. The results of our numerical
experiments are discussed in sections~\ref{ne2d} (the 2D case) and \ref{ne3d}
(the 3D case). In section~\ref{ne3d}, we describe necessary details of
the numerical scheme used in the 3D case. We draw our conclusions in
section~\ref{conclusion}.


\section{Interbank mutual obligations in a structural default model}

\label{stat}
Similar to \cite{LiptonSepp2009a,LiptonSepp2011} we consider a
multi-dimensional structural model inspired by the familiar model of \cite{m74},
see \cite{LiptonSepp2009a,LiptonSepp2011} and references therein.

First, for simplicity, assume that we have just two banks with external
assets $A_{i,t}$, $i=1,2$ and liabilities $L_{i,t}=G_{t}L_{i,0}$, and no
mutual liabilities. Here $G_{t}$ is the deterministic growth factor
\begin{equation}
G_{t}=\exp \left( \int_{0}^{t}r_{t^{\prime }}dt^{\prime }\right) ,
\label{Growth}
\end{equation}
\noindent where $r_{t}$ is the forward rate. Also assume that the default barrier
$l_{i,t}$ is a deterministic function of time\footnote{Below expression assumes that the bank assets are allowed to be below its liabilities up to some value determined by the recovery rate. In this case there is no default if such a breach is observed at some time before the maturity $T$. In this setup the default boundary has a kink at $t=T$.}:
\[ l_{i,t}=
\begin{cases}
R_{i}L_{i,t}, & t < T, \\
L_{i,T}, & t = T,
\end{cases}
\]
\noindent where $R_{i}$ is the average recovery of the bank's liabilities,
and $T$ is the debt maturity. Under normal circumstances, $R_{i}$ has a typical value
$R_{i}=0.4$.

We define the $i$th bank's default time $\tau _{i}$ assuming continuous
monitoring as follows
\[ \tau _{i}=\inf_{0<t\leq T}[A_{i,t}\leq l_{i,t}]. \]

Let us extend this approach by assuming that the banks in question do have
mutual liabilities, which we denote by $L_{ij,t}, \ i,j=1,2$; below we assume that $L_{ij,t} = G_t L_{ij,0}$. Thus, the total assets and liabilities of the $i$th bank are $A_{i}+\sum_{j\ne i}L_{ji}$ and $L_{i}+\sum_{j \ne i}L_{ij}$, respectively. Accordingly, the default time of the
first bank has the form
\begin{equation}
\tau _{1}=\inf_{0<t\leq T}[A_{1,t}\leq \lambda _{1,t}],  \label{def1}
\end{equation}
\noindent where
\[ \lambda _{1,t}=
\begin{cases}
R_{1}\left( L_{1,t}+L_{12,t}\right) -L_{21,t}, & t < T, \\
L_{1,T}+L_{12,T}-L_{21,T}, & t = T.
\end{cases}
\]
The default time of the second bank has a similar form
\begin{equation}
\tau _{2}=\inf_{0<t\leq T}[A_{2,t}\leq \lambda _{2,t}].  \label{def2}
\end{equation}

A new situation occurs, however, in case of default of one of the
banks. In case when the second bank defaults, it pays back to its creditors
only a portion of its liabilities, namely $R_{2}(L_{2}+L_{21})$. However,
the first bank pays back to the successors of the second bank the full
amount $L_{12}$, assuming of course that it does not default simultaneously
with the second bank. Thus, at time $\tau _{2}$ the first bank receives
from the second bank the amount $R_{2}L_{21}$ and pays the amount $L_{12}$. Therefore,
the new asset value $\tilde{A}_{1}$ of the first bank
becomes $\tilde{A}_{1}=A_{1}+R_{2}L_{21,\tau _{2}}-L_{12,\tau _{2}}$, while
its liability value becomes $L_{1,\tau _{2}}$. We assume that the actual
external assets do not jump in value, while the outside liabilities do get
adjusted. If the amount $R_{2}L_{21,\tau _{2}}-L_{12,\tau _{2}}$ is
positive, i.e., the first bank gets extra cash, which it spends retiring some of
the external liabilities. If this amount is negative, then it is borrowed
from the external sources. In both cases the total external liabilities become
\[
\tilde{L}_{1,\tau_2}=L_{1,\tau_2}-R_{2}L_{21,\tau_2}+L_{12,\tau_2}.
\]
Accordingly, the new default barrier for the first bank could be defined as
\[ \tilde{\lambda}_{1,t} =+
\begin{cases}
\lambda_{1,t}, & t < \tau_2, \\
\bar{\lambda}_{1,t}, & t \ge \tau_2,
\end{cases}
\qquad
\bar{\lambda}_{1,t} =
\begin{cases}
R_{1}\left( L_{1,\tau_2}-R_{2}L_{21,\tau_2}+L_{12,\tau_2}\right)G_t/G_{\tau_2}, & t < T, \\
(L_{1,\tau_2}-R_{2}L_{21,\tau_2}+L_{12,\tau_2})G_T/G_{\tau_2}, & t = T,
\end{cases}
\]
\noindent so that its default time has the form
\begin{equation}
\tilde{\tau}_{1}=\inf_{0<t\leq T}[A_{1,t}\leq \tilde{\lambda}_{1,t}].
\label{def1after}.
\end{equation}
It is easy to see, that after the default of the second bank, the default
boundary of the first bank increases by the amount of
\begin{equation}
\Delta \lambda _{1,\tau_2}=\tilde{\lambda}_{1,\tau_2}-\lambda _{1,\tau_2}=(1-R_{1}R_{2})L_{21,\tau_2}>0.
\label{boundChange1}
\end{equation}
Similarly,
\begin{equation}
\Delta \lambda _{2,\tau_1}=(1-R_{1}R_{2})L_{12,\tau_1}>0.  \label{boundChange2}
\end{equation}
Thus, the default boundary of the first bank jumps up by the increment $\Delta \lambda _{1}$
at time $\tau _{2}$, and the default boundary of the second
bank jumps up by the increment $\Delta \lambda _{2}$ at time $\tau _{1}$.
Mathematically, this means that our problem now has floating boundaries that
are deterministic functions of time which could increase at some moment by jumping to a higher value.
\begin{figure}[th]
\begin{center}
\fbox{
\begin{tikzpicture}
\def\sizeG{8.};
\def\lA{2.5};
\def\lB{3.};
\def\lAB{2.};
\def\lBA{1.5};
\draw (0,0) -- (\sizeG,0)
      (0,0) -- (0,\sizeG);
\draw[red, ultra thick] (0,\lB + \lBA) -- (\lA,\lB + \lBA) -- (\lA,\lB) -- (\sizeG,\lB);
\draw[blue, ultra thick] (\lA,\sizeG) -- (\lA,\lB) -- (\lB + \lAB,\lB) -- (\lB + \lAB,0);
\draw[red, ultra thick] (\lA,\lB) -- (\sizeG,\lB);
\draw[red, dashed] (0,\lB) -- (\lA,\lB);
\draw[blue, dashed] (\lA,\lB) -- (\lA,0);
\node at (\sizeG,-0.3){$A_1$};
\node at (-0.3,\sizeG){$A_2$};
\node at (-0.2,-0.2){$0$};
\node at (0.5*\lA,-0.3){$\lambda_1$};
\node at (-0.3,0.5*\lB){$\lambda_2$};
\node at (0.85*(\lA+\lAB,-0.3){$\Delta \lambda_1$};
\node at (-0.45,0.85*(\lB+\lBA){$\Delta \lambda_2$};
\node at (-0.3,\lB){$9$};
\node at (\lA + 0.3,\lB+0.3) {$3$};
\node at (\lA + \lAB + 0.3,\lB+0.3) {$7$};
\node at (\sizeG + 0.2,\lB+0.3) {$4$};
\node at (-0.3,\lB + \lBA) {$1$};
\node at (\lA + 0.3,\lB + \lBA) {$2$};
\node at (\lA + 0.3,0.3) {$6$};
\node at (\lA + \lAB + 0.3,0.3) {$8$};
\node at (\lA + \lAB + 0.3,0.3) {$8$};
\node at (\lA + 0.3,\sizeG) {$5$};
\end{tikzpicture}
}
\end{center}
\caption{Default boundaries of two banks with and without mutual liabilities.
}
\label{Fig1}
\end{figure}
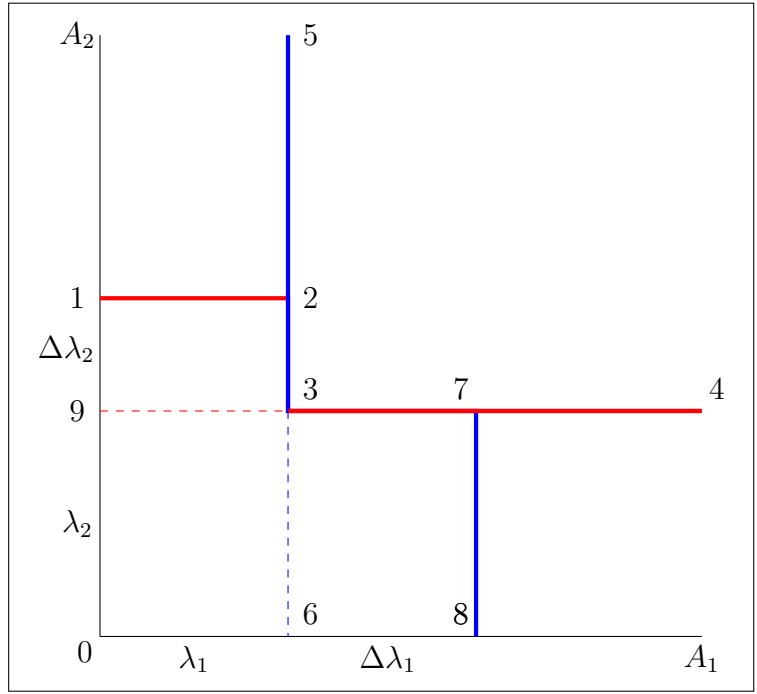

To illustrate the above observation, let us consider Fig.~\ref{Fig1} where
the situation is depicted at some moment of time $t < T$. If we don't take into
account mutual liabilities $L_{12}$ and $L_{21}$, then the default
boundaries are: for the first bank - a vertical line along the path
"5-2-3-6"; for the second bank - a horizontal line along the pass "9-3-7-4".
In the presence of mutual liabilities, the default boundary for the first
bank becomes "5-2-3-7-8", while for the second bank it has the form
"1-2-3-7-4".

A similar consideration can be used to show that the calculation of the marginal survival probabilities (which are needed to calibrate the model to the market CDS spreads) is strongly impacted by mutual liabilities. To emphasize this point, again consider the domain in Fig~\ref{Fig1}. Suppose we need to know $q_1(A_{1},A_{2},t,T)$ which is the marginal survival probability of the first bank conditional on the asset value $A_{2}$ of the second bank. In the presence of interbank liabilities we observe a new situation since the dynamics the first bank depends on the possible default of the second bank via the boundary conditions. Hence, the problem of computing $q_1$ remains inherently two-dimensional in contrast to the situation with no interbank liabilities.

In what follows we provide some numerical results that demonstrate this
behavior in the case of two and three firms by solving the corresponding 2D
and 3D PIDEs describing the evolution of both joint and survival probabilities in time and
space. We also discuss how parameters of the model affect the magnitude of
the effect.

To proceed further, we need to specify the dynamics of the external risky assets $A_{1,t}, A_{2,t}$; we assume that it could include both
diffusion and jumps components. We also assume that these assets are correlated as follows:

\begin{enumerate}
\item Diffusion components are correlated with the correlation coefficient $\rho$.

\item Jumps are correlated with the correlation
coefficient $\rho_{1,2}$ (see below for a more precise definition of this
correlation coefficient).

\item Changes in the firm value due to jumps and diffusion are uncorrelated.
\end{enumerate}

We assume that the underlying asset prices $A_{i,t}$ are driven by exponential
L{\'{e}}vy processes
\begin{equation}
A_{i,t}=A_{i,0}\exp (\Lambda _{i,t}),\quad 0\leq t\leq T,  \label{Levy}
\end{equation}
Under an appropriate pricing measure, each $\Lambda _{i,t}$ is characterized by a \LY triplet
$\left( \gamma_i, \sigma_i, \nu_i \right)$ with the drift $\gamma_i $, volatility $\sigma_i$, and L{\'{e}}vy measure $\nu_i$,
\begin{equation}
\Lambda _{i,t}=\gamma_{i}+\sigma _{i}W_{i,t}+Y_{i,t},\qquad \gamma_i ,\sigma_i \in
\mathbb{R},\quad \sigma_i >0,  \label{Lt}
\end{equation}
\noindent where $W_{t}$ is a standard Brownian motion on $0\leq t\leq T$ and
$Y_{t}$ is a pure jump process.\footnote{In order to better fit the market data, we can replace $\sigma_{i}$ with
the local volatility function $\sigma _{i}(t,A_{i,t})$.} We consider this
process under the pricing measure, therefore, $A_{i,t}/G_{t}$ is a
martingale. This allows us to express $\gamma $ as (\cite{Eberlein2009}) (further on we omit sub-index $i$ for simplicity)
\[
\gamma =r-\frac{\sigma ^{2}}{2}-\int_{\mathbb{R}}\left( e^{x}-1-x\mathbf{1}
_{|x|<1}\right) \nu (dx),
\]
with
\[
\int_{|x|>1}e^{x}\nu (dx)<\infty .
\]
At this stage, the jump measure $\nu (dx)$ is left unspecified, because we are open to
consider all types of jumps including those with finite and infinite
variation, and finite and infinite activity.

Let us introduce the logarithmic variables $x_{i}=\log A_{i}$ and define the
joint survival probability as follows
\[
Q(x_{1},x_{2},t,T)=\mathbf{1}_{\tau _{1}>t,\tau _{2}>t}E_{t}^{Q}[\mathbf{1}
_{\tau _{1}>T,\tau _{2}>T}].
\]
\noindent The joint survival probability solves the following PIDE, see \cite{LiptonSepp2011} and also
\cite{CliftForsyth2008}
\begin{equation}
Q_{\tau }=[\mathcal{L}+\mathcal{J}]Q  \label{PIDE},
\end{equation}
\noindent where $\tau =T-t$ is the backward time, and $\mathcal{L}$ is the two-dimensional linear convection-diffusion
operator of the form
\begin{equation}  \label{Dif}
\mathcal{L} = \sum_{i=1}^2 \left(r - \dfrac{1}{2}\sigma^2_i\right) \fp{}{x_i}
+ \dfrac{1}{2} \sum_{i,j=1}^2 \rho_{i,j} \sigma^2_i \sigma^2_i \fp{}{x_i}
\fp{}{x_j},
\end{equation}
\noindent and $\mathcal{J}$ is the jump operator
\begin{align}  \label{Jump}
\mathcal{J} Q = \int_{-\infty}^{\infty}\Big[Q(x_1&+y_1, x_2 + y_2, \tau)
- Q(x_1, x_2, \tau) - (e^{y_1} - 1)\fp{Q(x_1, x_2, \tau)}{x_1} \\
&- (e^{y_2} - 1)\fp{Q(x_1, x_2, \tau)}{x_2}\Big] \nu(d y_1 dy_2),
\nonumber
\end{align}
\noindent where $\nu(d y_1 d y_2)$ is the two-dimensional L{\'e}vy measure.

This PIDE has to be solved subject to the boundary and terminal conditions.
The terminal condition reads 
\[
Q(x_{1},x_{2},0)=\mathbf{1}_{x_{1}>\log \tilde{\lambda}_{1}(0),\
x_{2}>\log \tilde{\lambda}_{2}(0)}.
\]
The boundary conditions could be set as the Dirichlet conditions at $\pm
\infty $. Obviously,
\[
Q(x_{1},x_{2},\tau)\rightarrow 0,\ \mbox{at}\ x_{i}\rightarrow -\infty .
\]
As $x_{i}\rightarrow \infty ,\ i=1,2$, $Q(x_{1},x_{2},\tau)$ should
replicate the marginal survival probability $Q(x_{3-i},\tau)$. This
condition, however, must be supplemented with the boundary condition when
both $x_{1}\rightarrow \infty $ and $x_{2}\rightarrow \infty $. A natural
choice is $Q(x_{1}\rightarrow \infty ,x_{2}\rightarrow \infty ,\tau)=1$.

Various choices of the L\'evy measures that could be used for this model as
well as an approach to introduce the correlated jumps are discussed in the
next section.


\section{Correlated jumps and structured default models}

\label{secCorJumps}
There exist at least three known ways of introducing correlated jumps,
see \cite{ContTankov,Deelstra2010} and references therein.

The first one is to explicitly specify a multivariate distribution of the
jump process. This could be achieved, for instance, as in a celebrated Marshall-Olkin
paper (\cite{MarshallOlkin1967}) who use a multivariate exponential
distribution as a model for failure times, with the possibility of
simultaneous defaults. See also \cite{Linetsky2011} for the discussion of
this approach. The other possibility could be to use L\'evy copula, which in application
to the structural credit models was used by \cite{Baxter2007, Moosbrucker2006}. However, copula-based models impose some restrictive
constraints on the jump parameters to preserve marginal distributions, which
make it difficult to model arbitrary (positive and negative) correlations
between jumps. In other words, due to restrictions on the parameters
controlling marginal distributions, the correlation coefficient doesn't cover the entire range $[-1,1]$. The same problem is inherent in \cite{MarshallOlkin1967} construction as well, since this model
doesn't allow negative correlations between jumps, see, e.g., \cite{CliftForsyth2008}.

Another numerical approach to this problem has been established in \cite{Schwab2013}. The authors
develop Galerkin methods based on a wavelet-compression using the tensor
structure of the multi-dimensional PIDE operator to cope with the
complexity stemming from jumps as well as with the curse of dimensionality. The multivariate L\'evy processes in their framework
include jump diffusions and further allow for pure jump processes. The
correlation of the processes is constructed based on L\'evy copulas, see
also \cite{Schwab2004, Winter2009}. Accordingly, it is a subject of same
restrictions on the model parameters.

Another construction in \cite{LiptonSepp2009a} is also partly inspired by
the work of \cite{MarshallOlkin1967} with a significant advantage that both
positively and negatively correlated jumps can be represented.

The second approach uses multivariate subordinated Brownian
motions (or multivariate subordinators of Brownian motions), where the
L\'evy subordinator could consist of both common as well as idiosyncratic parts.
It is advocated by \cite{LucianoSemerato2010, Guillaume2013,
Linetsky2011}, see also survey in \cite{Ballotta2014} and references
therein. As applied to our problem it provides analytical
tractability if the local volatility is ignored. In this case the
characteristic function of the entire jump-diffusion model is known in
closed form, and transform methods, like FFT or cosine transform could be
used. With allowance for the local volatility this approach becomes
inefficient, because the jump integral must be computed at every point in
time and space.

In addition, this approach can only accommodate strictly
positive correlation values due to restrictions on the parameters controlling
the correlation coefficients. They are required to ensure the existence of
the characteristic function of the processes involved, see \cite{Ballotta2014}.

Therefore, we introduce the correlated jumps following the third approach \cite{Ballotta2014},
which constructs the jump process as a linear combination of two independent L\'evy processes representing the systematic factor and the idiosyncratic shock, respectively. Note, that
such an approach was also previously mentioned in \cite{ContTankov}. It has
an intuitive economic interpretation and retains nice tractability, as the
multivariate characteristic function in this model is available in closed
form.

The main result of \cite{Ballotta2014} that immediately follows from
Theorem~4.1 of \cite{ContTankov} (see also \cite{Garcia2009, Deelstra2010})
is given by:
\begin{proposition}
\label{prop0} Let $Z_t, \ Y_{j,t}, \ j=1,...,n$ be independent L\'evy
processes on a probability space $(Q, F, P )$, with characteristic functions
$\phi_Z(u; t)$ and $\phi_{Y_j}(u; t)$, for $j=1,...,n$ respectively. Then,
for $b_j \in \mathbb{R}, \ j=1,...,n$
\[
X_t = (X_{1,t},..., X_{n,t})^\top = (Y_{1,t} + b_1 Z_t,...,Y_{n,t} + b_n
Z_t)^\top
\]
is a L\'evy process on $\mathbb{R}^n$. The resulting characteristic function
is
\[
\phi_{\mathbf{X}}(\mathbf{u}; t) = \phi_Z\left( \sum_{i=1}^n b_i u_i;
t\right) \prod_{i=1}^n \phi_{Y_j}(u_j;t), \qquad \mathbf{u} \in \mathbb{R}^n.
\]
\end{proposition}

By construction every factor $X_{i,t}, i=1,...,n$ includes a common factor $Z_t$. Therefore, all components $X_{i,t}, i=1,...,n$ could jump together,
and loading factors $b_i$ determine the magnitude (intensity) of the jump in
$X_{i,t}$ due to the jump in $Z_t$. Thus, all components of the multivariate
L\'evy process $\mathbf{X}_t$ are dependent, and their pairwise correlation is given by (again
see \cite{Ballotta2014} and references therein)
\[
\rho_{j,i} = \mbox{Corr}(X_{j,t}, X_{i,t}) = \dfrac{b_j b_i \mbox{Var}(Z_1)} {\sqrt{\mbox{Var}(X_{j,1})}\sqrt{\mbox{Var}(X_{j,1})}}
\]
Such a construction has multiple advantages, namely:

\begin{enumerate}
\item As $\mbox{sign}(\rho_{i,j}) = \mbox{sign}(b_i b_j)$, both positive and
negative correlations can be accommodated

\item In the limiting case $b_i \to 0$ or $b_j \to 0$ or $\mbox{Var}(Z_1) =
0 $ the margins become independent, and $\rho_{i,j} = 0$. The other limit $b_i \to \infty$ or $b_j \to \infty$ represents a full positive correlation
case, so $\rho_{i,j} = 1$. Accordingly, $b_i \to \infty, \ b_{3-i} \to
\infty, \ i=1,2$ represents a full negative correlation case as in this
limit $\rho_{i,j} = -1$.
\end{enumerate}

One more advantage of this approach becomes apparent if we want the margin
distribution $X_{i,t}$ to be fixed. Then a set of conditions on convolution
coefficients could be imposed to preserve the margin. This is reasonable
from the practical viewpoint as the entire credit product could be
illiquid, and, therefore, the market quotes necessary to calibrate the full
correlation matrix might not be available. Hence, as an alternative, the marginal
distributions could be first calibrated to a more liquid market of the
components $X_{i,t}$, and the entire correlation structure should preserve
these marginals. As a first step, this defines parameters of the idiosyncratic
factors. As the next step, the remaining parameters of the entire correlation structure
are, based on a separate consideration. Note, that a similar
idea is used in another recent paper \cite{Mai2014}, where the authors
concentrate on two specific models for the marginals, and achieve
tractability by choosing the relevant parameters in such a way that univariate
marginals are separated from dependence structure. However, in the present
approach, any model could be treated in a unified way.

According to this setup, the instantaneous correlation between the log-assets $x_1$ and $x_2$ reads
\begin{equation}  \label{corr}
\rho_{12} = \dfrac{\rho \sigma_1 \sigma_2 + b_1 b_2 \mbox{Var}(Z_1)}{\sqrt{\sigma_1^2 + \mbox{Var}(X_{1,1})} \sqrt{\sigma_2^2 + \mbox{Var}(X_{2,1})}}
\end{equation}

As far as the structural default model is concerned, positive jumps might
not be necessary. However, below we keep them for generality,
as the proposed approach to modeling correlated jumps is applicable without any modification
in other settings, where both positive and negative jumps are important.


\section{Fractional PDE and jump integrals}

\label{Frac} 
Assuming that some particular L\'evy models are chosen to construct
processes $Y_{i,t}, i=1,...,n$ and $Z_t$, let us look more closely at
\eqref{Jump}. In doing that we follow the method proposed in \cite{ItkinCarr2012Kinky} (first presented at Global Derivatives and Risk
Conference, Roma 2009) and then further elaborated on in \cite{Itkin2014,
Itkin2014a}. The key idea is to represent the jump integral in the form of a
pseudo-differential operator and then formally solve, thus obtained
evolutionary partial pseudo-differential equation via a matrix exponential.

To be clear, we start with a one-dimensional case. It is
well known from quantum mechanics \cite{OMQM} that a translation (shift)
operator in $L_2$ space could be represented as
\begin{equation}  \label{transform}
\mathcal{T}_b = \exp \left( b \dfrac{\partial}{\partial x} \right),
\end{equation}
with $b$ = const, so
\[
\mathcal{T}_b f(x) = f(x+b).
\]

Therefore, the one-dimensional integral corresponding to \eqref{Jump} can be
formally rewritten as
\begin{align}  \label{intGen}
\int_\mathbb{R} \left[ Q(x+y,t) \right. & \left. - Q(x,t) - (e^y-1)
\fp{Q(x,t)}{x} \right] \nu(dy) = \mathcal{J} Q(x,t), \\
\mathcal{J} & \equiv \int_\mathbb{R}\left[ \exp \left( y \dfrac{\partial}{
\partial x} \right) - 1 - (e^y-1) \fp{}{x} \right] \nu(dy).  \nonumber
\end{align}

In the definition of the operator $\mathcal{J}$ (which is actually an
infinitesimal generator of the jump process), the integral can be formally
computed under some mild assumptions about existence and convergence if one
treats the term $\partial/ \partial x$ as a constant. Therefore, the operator $\mathcal{J}$ can be considered as some generalized function of the
differential operator $\partial_x$. We can also treat $\mathcal{J}$ as a
pseudo-differential operator.

It is important to emphasize that
\begin{equation}  \label{MGF}
\mathcal{J} = \psi(-i \partial_x) - [\log \psi(-i)]\partial_x = \mbox{MGF}
(\partial_x) - [\log \mbox{MGF}(1)]\partial_x,
\end{equation}
\noindent where $\psi(u)$ is the characteristic exponent of the jump
process, and $\mbox{MGF(u)}$ is the moment generation function corresponding
to this characteristic exponent. This directly follows from the L{\'e}vy-Khinchine theorem. Note, that the last term on the right hand side of \eqref{MGF} is a compensator as the characteristic exponent is computed
using the expectation under a risk-neutral measure $\mathbb{Q}$. In other
words, the last term is added to make the forward price to be a true
martingale under this measure.

This representation is advantageous because it transforms a linear non-local Integro-differential operator (jump operator)
into a linear local pseudo-differential (fractional) operator. The operator $\mathcal{J}$ can be analytically computed for various popular L\'evy models,
hence $\mathcal{J}$ admits an explicit representation in the form of the
pseudo-differential operator. Accordingly, a pure jump evolutionary equation
\[
Q_{\tau} = \mathcal{J} Q
\]
could be formally integrated (under some mild existence conditions) to provide
\[
Q(x,\tau + \Delta \tau) = e^{\Delta \tau\mathcal{J}} Q(z,\tau).
\]
The operator $\mathcal{A} = e^{\Delta \tau\mathcal{J}}$ is the matrix
exponential and is understood as a Taylor series expansion of $\Delta \tau \mathcal{J}$.

In \cite{ItkinCarr2012Kinky, Itkin2014, Itkin2014a} it is shown that the
matrix exponential can be efficiently computed on a finite difference grid
for various jump models, namely Merton, Kou, CGMY, NIG, General Hyperbolic
and Meixner models. Efficiency of this method in general is not worse than
that of the FFT, and in many cases is linear in $N$ - the number of the grid
points\footnote{In particular, this is the case for the Merton, Kou, CGMY and Meixner models.
In this paper we also prove it for the exponential L\'evy model which is a
particular case of the Kou double-exponential model, see Appendix~\ref{ap4}.}. The proposed method is almost universal, i.e., it allows solving PIDEs for
various jump-diffusion models in a unified form. Second order finite
difference schemes in both space and time are constructed in such a way that i) they are
unconditionally stable, and ii) they preserve positivity of the solution.
Therefore, we assume this method to be robust and more efficient than
constructions proposed in the literature to solve a similar class of
problems, e.g., Galerkin methods of \cite{Schwab2013} which even for sparse
matrices don't reach the linear complexity in each dimension. In addition,
the construction of the correlated jumps using the L\'evy copulas used in
\cite{Schwab2013} is restrictive as this was already discussed in Section~\ref{secCorJumps}.

Now let us use the same idea for getting fractional representation of the
jump integral in the two-dimensional case. The translational two-dimensional
operator in $L_2 \times L_2$ space could be similarly represented as
\begin{equation}  \label{transform2}
\mathcal{T}_{y_1,y_2} = \exp \left( y_1 \dfrac{\partial}{\partial x_1}
\right)\exp \left( y_2 \dfrac{\partial}{\partial x_2} \right),
\end{equation}
with $y_1, y_2$ = const, so
\[
Q(x_1 + y_1, x_2 + y_2, \tau) = \mathcal{T}_{y_1,y_2} Q(x_1,x_2,\tau)
\]
Therefore, the whole integral in \eqref{Jump} could be re-written in the
form
\begin{equation}  \label{Jump2d}
\mathcal{J} = \int_{-\infty}^{\infty}\Big[ e^{y_1 \partial_{x_1}} e^{y_2
\partial_{x_2}} - 1 - (e^{y_1} - 1)\partial_{x_1} - (e^{y_2} -
1)\partial_{x_2}\Big] \nu(d y_1 d y_2).
\end{equation}
By using Proposition~\ref{prop0} and the L\'evy-Khinchine
theorem, similar to how the \eqref{MGF} was derived, we can show
that
\begin{equation}  \label{JumpPhi}
\mathcal{J} = \sum_{j=1}^2 \psi_{X_j}(-i \partial_{x_j}) +
\psi_Z \left(-i \sum_{j=1}^2 b_j \partial_{x_j} \right) + 1 - \sum_{j=1}^2
[\log \psi_{X_j}(-i)]\partial_{x_j},
\end{equation}
Based on \cite{ItkinCarr2012Kinky, Itkin2014, Itkin2014a} we know how to
deal with all the terms in this expression except the new term $\psi_Z$
which represents a two-dimensional characteristic exponent of the common
jump process $Z_t$. We shall discuss this in the next sections.


\section{Splitting on financial processes}

\label{splitting}
To solve \eqref{PIDE} we use an FD approach with splitting in
financial processes. We refer the reader to \cite{Itkin2014} to the detailed
description of the splitting algorithm. Splitting (a.k.a. the method of
fractional steps) reduces the solution of the original k-dimensional
unsteady problem to the solution of $k$ one-dimensional equations per time
step. For example, consider a two-dimensional diffusion equation with a
solution obtained by using some FD method. At every time step, a standard discretization in space variables is applied, such that the FD grid contains $N_1$ nodes in the first dimension and $N_2$ nodes in the second dimension. Then the problem reduces to solving a system of $N_1 \times N_2$ linear equations with a block-diagonal matrix. In contrast, utilization of splitting results in, e.g., $N_1$ systems of $N_2$ linear equations, where the matrix of each system is banded (tridiagonal). The latter approach is easy to implement and, more importantly, provides significantly better performance.

A natural choice for the first step would be to split operators
$\mathcal{L}$ and $\mathcal{J}$ in \eqref{PIDE} separately due to their different
mathematical nature. So a special scheme could be applied at every step of
the splitting procedure. As operators $\mathcal{L}$ and $\mathcal{J}$ are
non-commuting, we use Strang's splitting scheme, \cite{Strang1968}, which
provides second order approximation in time $\tau$ assuming that at every
step of splitting the corresponding equations are solved also with the
second order accuracy in time. For more details on how to apply Strang's
splitting to fractional equations see \cite{Itkin2014} and references
therein. The entire numerical scheme reads
\begin{align}  \label{splitFin}
Q^{(1)}(x_1,x_2,\tau) &= e^{\frac{\Delta \tau}{2} \mathcal{D} }
Q(x_1,x_2,\tau), \\
Q^{(2)}(x_1,x_2,\tau) &= e^{\Delta \tau \mathcal{J}} Q^{(1)}(x_1,x_2,\tau) ,
\nonumber \\
Q(x_1,x_2,\tau + \Delta \tau) &= e^{\frac{\Delta \tau}{2} \mathcal{D} }
Q^{(2)}(x_1,x_2,\tau).  \nonumber
\end{align}

Thus, instead of an non-stationary PIDE, we obtain one PIDE with no drift and no re-wri
diffusion (the second equation in \eqref{splitFin}) and two non-stationary PDEs
(the first and third ones in \eqref{splitFin})\footnote{As we use splitting on financial processes, pure jump models are naturally
covered by the same method. In the latter case there is no diffusion at the
first and third step of the method, so one has to solve a pure convection
equation. This could be achieved by applying various methods known in the
fluid mechanics literature, see, e.g., \cite{Roach1976}.}.

Proceeding in a similar way, the second step is to apply splitting to the
second equation in \eqref{splitFin}. We represent \eqref{JumpPhi} in the form
\begin{equation}  \label{split2}
\mathcal{J} = \mathcal{J}_1 + \mathcal{J}_2 + \mathcal{J}_{12},
\end{equation}
\noindent where
\begin{align*}
\mathcal{J}_j &= \psi_{X_j}(-i \partial_{x_j}) - [\log \psi_{X_j}(-i)]\partial_{x_j}, \qquad j=1,2 \\
\mathcal{J}_{12} &= 1 + \psi_Z \left(-i \sum_{j=1}^2 b_j \partial_{x_j} \right).
\end{align*}
Obviously, operators $\mathcal{J}_1$ and $\mathcal{J}_2$ commute, so that
\[
e^{t(\mathcal{J}_1 + \mathcal{J}_2)} = e^{t \mathcal{J}_1}e^{t\mathcal{J}_2}
\]
Therefore, replacing the second step in \eqref{splitFin} with another
Strang's splitting using \eqref{split2}, we finally obtain
\begin{align}  \label{splitAll}
Q^{(1)}(x_1,x_2,\tau) &= e^{\frac{\Delta \tau}{2} \mathcal{D} }
Q(x_1,x_2,\tau), \\
Q^{(2)}(x_1,x_2,\tau) &= e^{\frac{\Delta \tau}{2} \mathcal{J}_1}
Q^{(1)}(x_1,x_2,\tau) ,  \nonumber \\
Q^{(3)}(x_1,x_2,\tau) &= e^{\frac{\Delta \tau}{2} \mathcal{J}_2}
Q^{(2)}(x_1,x_2,\tau) ,  \nonumber \\
Q^{(4)}(x_1,x_2,\tau) &= e^{\Delta \tau \mathcal{J}_{12}}
Q^{(3)}(x_1,x_2,\tau) ,  \nonumber \\
Q^{(5)}(x_1,x_2,\tau) &= e^{\frac{\Delta \tau}{2} \mathcal{J}_2}
Q^{(1)}(x_1,x_2,\tau) ,  \nonumber \\
Q^{(6)}(x_1,x_2,\tau) &= e^{\frac{\Delta \tau}{2} \mathcal{J}_1}
Q^{(1)}(x_1,x_2,\tau) ,  \nonumber \\
Q(x_1,x_2,\tau + \Delta \tau) &= e^{\frac{\Delta \tau}{2} \mathcal{D} }
Q^{(6)}(x_1,x_2,\tau).  \nonumber
\end{align}

\section{Numerical procedure}

\label{NP} Due to the splitting nature of our entire algorithm represented
by \eqref{splitAll}, each step of splitting is computed using a separate
numerical scheme. All schemes provide second order approximation in both
space and time, are unconditionally stable and preserve positivity of the
solution.

For the first and the last step where a pure convection-diffusion
two-dimensional problem has to be solved we use a Hundsdorfer-Verwer scheme,
see \cite{HoutWelfert2007, HoutFoulon2010, Itkin2014b}. A non-uniform
finite-difference grid is constructed similar to \cite{ItkinCarrBarrierR3}.

For the steps 2,3,5,6 we choose the Merton jump model. In other words, the
idiosyncratic jump part of each component $X_{j,t}, \ j=1,2$ is represented
as Gaussian. Computation of the matrix exponential $\mathcal{A}_j
Q(x_1,x_2,\tau) = e^{\frac{\Delta \tau}{2} \mathcal{J}_j} Q(x_1,x_2,\tau), \
j=1,2$ could be done with complexity $O(N_1 N_2)$ at every time step. This
is because when computing $\mathcal{A}_1$ the second variable $x_2$ is a
dummy variable, while computation of $\mathcal{A}_1 Q(x_1,x_2 = const,\tau)$
is $O(N_1)$, see \cite{Itkin2014}. Construction of the jump grid, which is a
superset of the finite-difference grid used at the first (diffusion) step is
also described in detail in \cite{Itkin2014}.

For step 4 (common or systemic jumps) we choose the Kou double exponential jumps
model proposed in \cite{Kou2004}. Its \LY density is
\begin{equation}  \label{Kou}
\nu (dx) = \varphi\left[ p \theta_1 e^{-\theta_1 x} \mathbf{1}_{x \ge 0} +
(1-p) \theta_2 e^{\theta_2 x} \mathbf{1}_{x < 0} \right] dx,
\end{equation}
where $\varphi$ is the jumps intensity, $\theta_1 > 1$, $\theta_2 > 0$, $1 > p > 0$; the first condition was
imposed to ensure that the underlying asset price has a finite expectation.

Using this model a one-dimensional representation for $\mathcal{J}$ is given
in \cite{Itkin2014}. Similarly, in a two dimensional case we obtain
\begin{align}  \label{KouJ}
\mathcal{J}_{12} &= \varphi \left[p \theta_1(\theta_1- b_1
\triangledown_{x_1} - b_2 \triangledown_{x_2})^{-1} + (1-p)
\theta_2(\theta_2 + b_1 \triangledown_{x_1} + b_2 \triangledown_{x_2})^{-1}
\right], \\
\triangledown_{x_1} &\equiv \partial_{x_1}, \quad \triangledown_{x_2} \equiv
\partial_{x_2}, \quad -\theta_2 < Re(b_1 \triangledown_{x_1} + b_2
\triangledown_{x_2}) < \theta_1.  \nonumber
\end{align}
The inequality $-\theta_2 < Re(\triangledown) < \theta_1$ is an existence
condition for the integral defining $\mathcal{J}$ and should be treated as
follows: the discretization of the operator $\triangledown$ should be such
that all eigenvalues of matrix $A$, a discrete analog of $\triangledown$,
obey this condition.

We proceed in a way similar to the one-dimensional case. To this end we can
use the (1,1) P{\'a}de approximation of $e^{\Delta \tau \mathcal{J}_{12}}$
which provides $O( (\Delta \tau)^2)$ approximation of the form
\begin{equation}  \label{pade}
e^{\Delta \tau \mathcal{J}_{12}} \approx [1 - \frac{1}{2}\Delta \tau
\mathcal{J}_{12}]^{-1}[1 + \frac{1}{2}\Delta \tau \mathcal{J}_{12}] +
O(\Delta \tau^3).
\end{equation}
This scheme can also be re-written as
\begin{equation}  \label{Picard}
Q(x_1, x_2, \tau + \Delta \tau) - Q(x_1, x_2, \tau) = \frac{1}{2} \Delta
\tau \mathcal{J}_{12} \left[Q(x_1,x_2, \tau + \Delta \tau) + Q(x_1, x_2,\tau)
\right],
\end{equation}
\noindent and this equation could be solved using the Picard fixed-point
iterations. In doing so, we observe that the entire product $\mathcal{J}_{12}
Q(x_1,x_2,\tau)$ with $\mathcal{J}_{12} $ given in \eqref{KouJ} can be
calculated as follows.

\subsubsection*{First term.}

Observe that the vector $z(x_1,x_2,\tau) = p \theta_1 (\theta_1- b_1
\triangledown_{x_1} - b_2 \triangledown_{x_2})^{-1} Q(x_1, x_2,\tau)$ solves
the equation
\begin{equation}  \label{kou1system}
(\theta_1- b_1 \triangledown_{x_1} - b_2 \triangledown_{x_2})
z(x_1,x_2,\tau) = p \theta_1 Q(x_1,x_2,\tau).
\end{equation}
This is a two-dimensional linear PDE of the first order. It could be solved
numerically with the second order approximation in $x_1,x_2$ using the
Peaceman-Rachford ADI method, see \cite{McDonough2008}
\begin{align}  \label{adi1}
\left[\left(s + \dfrac{1}{2} \theta_1 \right)- b_1 \triangledown_{x_1}\right]
z^*(x_1,x_2,\tau) &= \left[\left(s - \dfrac{1}{2} \theta_1 \right)+ b_2
\triangledown_{x_2}\right] z^k(x_1,x_2,\tau) + b \\
\left[\left(s + \dfrac{1}{2} \theta_1 \right)- b_2 \triangledown_{x_2}\right]
z^{k+1}(x_1,x_2,\tau) &= \left[\left(s - \dfrac{1}{2} \theta_1 \right)+ b_1
\triangledown_{x_1}\right]z^*(x_1,x_2,\tau) + b  \nonumber \\
b &\equiv p \theta_1 Q(x_1,x_2,\tau).  \nonumber
\end{align}
Here $s > 0$ is some parameter that could be chosen in a special way to
provide convergence of the method, see Appendices. The number $k$ is
the iteration number, the whole process starts with $k=1$.

Before constructing a finite difference scheme to solve this equation we
need to introduce some definitions. Define a one-sided \textit{forward}
discretization of $\triangledown$, which we denote as $A^F: A^F C(x) =
[C(x+h,t) - C(x,t)]/h$. Also define a one-sided \textit{backward}
discretization of $\triangledown$, denoted as $A^B: \ A^B C(x) = [C(x,t) -
C(x-h,t)]/h$. These discretizations provide first order approximation in $h$, e.g., $\triangledown C(x) = A^F C(x) + O(h)$. To provide the second order
approximations, use the following definitions. Define $A^C_2 = A^F\dot A^B$
- the \textit{central} difference approximation of the second derivative $\triangledown^2$, and $A^C = (A^F + A^B)/2$ - the \textit{central}
difference approximation of the first derivative $\triangledown$. Also
define a one-sided second order approximations to the first derivatives:
\textit{backward} approximation $A^B_2: \ A^B_2 C(x) = [ 3 C(x) - 4 C(x-h) +
C(x-2h)]/(2 h)$, and \textit{forward} approximation $A^F_2: \ A^F_2 C(x) = [
-3 C(x) + 4 C(x+h) - C(x-2h)]/(2 h)$. Also $I$ denotes a unit matrix. All
these definitions assume that we work on a uniform grid, however this could
be easily generalized for the non-uniform grid as well, see, e.g., \cite{HoutFoulon2010}.

The following Proposition now solves the problem \eqref{adi1}

\begin{proposition}
\label{PropKou1} Consider the following discrete approximation of the ADI
scheme \eqref{adi1}:
\begin{align}  \label{adi1d}
\Big[\Big(s &+ \dfrac{1}{2} \theta_1 \Big)I_{x_1} - b_1 A(x_1)\Big]
z^*(x_1,x_2,\tau) = \Big[\Big(s - \dfrac{1}{2} \theta_1 \Big)I_{x_2} + b_2
A(x_2)\Big]z^k(x_1,x_2,\tau) + b \\
\Big[\Big(s &+ \dfrac{1}{2} \theta_1 \Big)I_{x_2} - b_2 A(x_2)\Big]
z^{k+1}(x_1,x_2,\tau) = \Big[\Big(s - \dfrac{1}{2} \theta_1 \Big)I_{x_1} +
b_1 A(x_1)\Big]z^*(x_1,x_2,\tau) + b  \nonumber \\
b &\equiv p \theta_1 Q(x_1,x_2,\tau), \qquad A(x_i) =
\begin{cases}
A_2^F(x_i), & b_i > 0 \\ A_2^B(x_i), & b_i < 0, \qquad i=1,2
\end{cases}
\nonumber
\end{align}

Then this scheme is unconditionally stable, approximates the original PDE \eqref{adi1} with accuracy $O( (\Delta x_1)^2 + (\Delta x_2)^2 + (\Delta x_1)(\Delta
x_2))$ and preserves positivity of the solution.
\end{proposition}

\begin{proof}
See Appendix~\ref{ap1}.
\end{proof}
We can start iterations in \eqref{adi1d} by choosing $z^{(1)}(x_1, x_2,
\tau) = Q(x_1, x_2, \tau)$. In our experiments the scheme converges to the
solution after 5-6 iterations if we choose $s = \theta_1 + 1$ in \eqref{adi1} and $s = \theta_2 + 1$ in \eqref{adi2}.

\subsubsection*{Second term.}

Observe that the vector $z(x_1,x_2,\tau) = (1-p) \theta_2 (\theta_2 + b_1
\triangledown_{x_1} + b_2 \triangledown_{x_2})^{-1} Q(x_1, x_2,\tau)$ solves
the equation
\begin{equation}  \label{kou2system}
(\theta_2 + b_1 \triangledown_{x_1} + b_2 \triangledown_{x_2})
z(x_1,x_2,\tau) = (1-p) \theta_2 Q(x_1,x_2,\tau).
\end{equation}
This is also a two-dimensional linear PDE of the first order, so again we
apply the Peaceman-Rachford method
\begin{align}  \label{adi2}
\left[\left(s + \dfrac{1}{2} \theta_2 \right)+ b_1 \triangledown_{x_1}\right]
z^*(x_1,x_2,\tau) &= \left[\left(s - \dfrac{1}{2} \theta_2 \right)- b_2
\triangledown_{x_2}\right](1-p) z^k(x_1,x_2,\tau) + b \\
\left[\left(s + \dfrac{1}{2} \theta_2 \right)+ b_2 \triangledown_{x_2}\right]
z^{k+1}(x_1,x_2,\tau) &= \left[\left(s - \dfrac{1}{2} \theta_2 \right)- b_1
\triangledown_{x_1}\right]z^*(x_1,x_2,\tau) + b  \nonumber \\
b &\equiv (1-p)\theta_2 Q(x_1,x_2,\tau)  \nonumber
\end{align}

The next Proposition provides a construction of the finite difference scheme to
solve the problem \eqref{adi2}

\begin{proposition}
\label{PropKou2} Consider the following discrete approximation of the ADI
scheme \eqref{adi2}:
\begin{align}  \label{adi2d}
\Big[\Big(s &+ \dfrac{1}{2} \theta_2 \Big)I_{x_1} + b_1 A(x_1)\Big]
z^*(x_1,x_2,\tau) = \Big[\Big(s - \dfrac{1}{2} \theta_2 \Big)I_{x_2} - b_2
A(x_2)\Big]z^k(x_1,x_2,\tau) + b \\
\Big[\Big(s &+ \dfrac{1}{2} \theta_2 \Big)I_{x_2} + b_2 A(x_2)\Big]
z^{k+1}(x_1,x_2,\tau) = \Big[\Big(s - \dfrac{1}{2} \theta_2 \Big)I_{x_1} -
b_1 A(x_1)\Big]z^*(x_1,x_2,\tau) + b  \nonumber \\
b &\equiv (1-p) \theta_2 Q(x_1,x_2,\tau), \qquad A(x_i) = \begin{cases}
A_2^B(x_i), & b_i > 0 \\ A_2^F(x_i), & b_i < 0, \qquad i=1,2\end{cases}
\nonumber
\end{align}
Then this scheme is unconditionally stable, approximates the original PDE
\eqref{adi2} with $O( (\Delta x_1)^2 + (\Delta x_2)^2 + (\Delta x_1)(\Delta
x_2))$ and preserves positivity of the solution.
\end{proposition}

\begin{proof}
See Appendix~\ref{ap2}.
\end{proof}
Overall, our experiments show that the first Picard scheme \eqref{Picard}
converges after 2-3 iterations to the absolute accuracy of $2\cdot 10^{-4}$.

To summarize, the total complexity of the proposed splitting algorithm at
every time step is $O(\alpha N_1 N_2)$, where $\alpha$ is some constant
coefficient. To estimate it, observe that the solution of the
convection-diffusion equation requires five sweeps, where at every sweep
either $N_1$ systems of linear equations with the tridiagonal matrix of size
$N_2$, or $N_2$ systems of size $N_1$ have to be solved, see \cite{HoutFoulon2010}. The idiosyncratic jump parts modeled by the Merton jump
model are solved with the complexity $O(N_1 N_2)$ (i.e., at this step $\alpha = 1$) using the Improved Fast Gauss Transform (IFGT), see \cite{Itkin2014}. As we need to provide two steps of splitting in the $x_1$
dimension, and two other steps in the $x_2$, the total number of sweeps is
four. Finally the above algorithm for computing common jumps using the Kou
model requires 2-3 Picard iterations for the matrix exponential, and at
every iteration we solve 2 ADI systems of linear equations using 5-6
iterations, so in total about 30 sweeps. Thus, overall $\alpha$ is about 44.
This is still better than a straightforward application of the FFT which
usually requires the number of FFT nodes to be a power of 2 with a typical
value of 2$^{11}$. It is also better than the traditional approach which
considers approximation of the linear non-local jump integral $\mathcal{J}$
on some grid and then makes use of the FFT to compute a matrix-by-vector
product. Indeed, when using FFT for this purpose we need two sweeps per
dimension using a slightly extended grid (with, say, the tension coefficient
$\xi$) to avoid wrap-around effects, \cite{Halluin2005b}. Therefore the
total complexity per time step could be at least $O( 4 \xi_1 \xi_2 N_1 N_2
\log_2 (\xi_1 N_1) \log_2 (\xi_2 N_2))$ which even for a relatively small
FFT grid with $N_1 = N_2 = 512$, and $\xi_1 = \xi_2 = 1.1$ is about 9 times
slower than our method. Also traditional approach experiences some other
problems for jumps with infinite activity and infinite variation, see survey
in \cite{Itkin2014} and references therein.

If instead of the Kou model one wants to apply the Merton jump model for
systemic jumps, it becomes a bit more computationally expensive. Indeed, at
every time step the multi-dimensional diffusion equation with constant
coefficients could be effectively solved by using the IFGT. Suppose, in doing
so, we want to achieve the accuracy $10^{-3}$. Then, roughly, we need to keep
$p=9$ terms in the Taylor series expansion of IFGT, and the total complexity
for the two-dimensional case $d=2$ is $O(90 N_1 N_2)$, see \cite{IFGT}.

\section{Numerical experiments} \label{ne2d}

\subsection{The one-dimensional problem}
We start with the one-dimensional model for two reasons. First, the solution of this model is used as the boundary condition for the two-dimensional problem. Second, in some cases, e.g., for the exponential jumps, this model could be solved in the closed form, and, therefore, can be utilized for verification of the method.

In the first test we consider the one-dimensional pure diffusion problem. We solve it as a limiting  case of the two-dimensional problem when the volatility and drift of the second asset vanish. This solution for the survival probability is compared with the analytical solution which in this case coincides with the price of a digital down-and-out call option, see \cite{Howison1995}. Thus, in this test the robustness of our convection-diffusion FD scheme is validated. Parameters of the model used in this test are given in Table~\ref{Tab1dD}, and the results are presented in Fig.~\ref{dif1dD} where the absolute value of the relative difference between the analytical price and one computed by our finite-difference method is depicted as a function of $A_{1,0}$. As is shown in this Figure the relative error is below 1\% everywhere except in the close vicinity of the barrier where the value of $Q$ itself is small.
\begin{table}[!ht]
\begin{center}
\begin{tabular}{|c|c|c|c|c|c|c|c|c|c|c|c|c|}
\hline
$A_{2,0}$ & $L_{1,0}$ & $L_{2,0}$ & $L_{12,0}$ & $L_{21,0}$ & $R_{1}$ & $R_{2}$ & $r$ &  T & $\sigma_1$ & $\sigma_2$ \cr
\hline
100 & 40 & 0 & 0 & 0 & 1 & 0 & 0.05 & 1 & 0.2 & 0 \cr
\hline
\end{tabular}
\caption{Parameters of the structural 1D default model.}
\label{Tab1dD}
\end{center}
\end{table}

\begin{figure}[!ht]
\begin{center}
\fbox{\includegraphics[width=4 in]{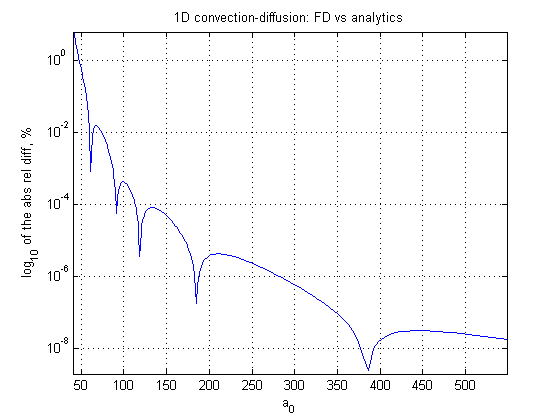}}
\caption{The absolute value of the relative difference between the analytical and the FD solutions
for the convection-diffusion problem as a function of the initial asset value $A_{1,0}$.}
\label{dif1dD}
\end{center}
\end{figure}

In the second test we extend the previous case by adding exponential jumps to the first component. Again, this problem admits an analytical solution which could be expressed via the inverse Laplace transform, see \cite{Lipton2002a} where this problem was solved by using fluctuation identities. It can also be solved by using a generalized transform of \cite{Lewis:2000} combined with  the Wiener-Hopf method, see, e.g., \cite{Kuznetsov2011}. The corresponding solution reads
\[ Q = 1 - \mathcal{L}_q^{-1}\left\{\dfrac{e^{\beta \hat{\rho}_1(q)}}{q}\right\}. \]
Here $\hat{\rho}_1$ is the only negative root of the characteristic equation in the Wiener-Hopf method, $\beta = \log(B_a/a_0) < 0$, and $B_a$ is the default boundary.

Also within the framework of \cite{Itkin2014a} which we use in this paper, exponential jumps were never considered. Therefore, in Appendix~\ref{ap4} for completeness, we construct a finite-difference algorithm for exponential jumps.

In Fig.~\ref{dif1dJ} the absolute value of the relative difference between the analytical and numerical solutions is depicted as a function of $A_{1,0}$. In this experiment we set the intensity of the jumps $\lambda = 0.7$, and the parameter of the exponential distribution $\phi = 2$.
\begin{figure}[!ht]
\begin{center}
\fbox{\includegraphics[width=4 in]{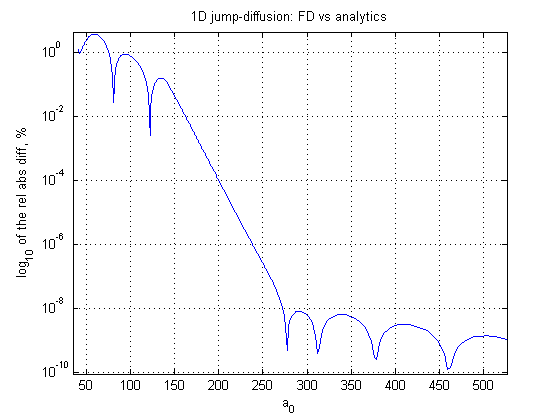}}
\caption{The absolute value of the relative difference between the analytical and the FD solutions
for the jump-diffusion problem as a function of the initial asset value $A_{1,0}$.}
\label{dif1dJ}
\end{center}
\end{figure}
The difference is less than 1\% except close to the barrier; see also Table~\ref{TabComp}.
\begin{table}[!ht]
\begin{center}
\begin{tabular}{|c|c|c|c|c|c|c|c|c|c|c|c|c|c|}
\hline
$A_{1,0}$ & $Q_{an}$ & $Q_{FD}$ & $\Delta Q$ \cr
\hline
40.85  &  0.008805  &   0.008909  &   -0.000103 \cr
41.69  &  0.017710  &   0.017874  &   -0.000164 \cr
42.53  &  0.026710  &   0.026972  &   -0.000262 \cr
43.36  &  0.035802  &   0.036188  &   -0.000387 \cr
44.18  &  0.044982  &   0.045523  &   -0.000541 \cr
44.99  &  0.054251  &   0.054981  &   -0.000729 \cr
45.79  &  0.063610  &   0.064563  &   -0.000953 \cr
46.59  &  0.073058  &   0.074273  &   -0.001215 \cr
47.38  &  0.082601  &   0.084116  &   -0.001515 \cr
48.16  &  0.092241  &   0.094094  &   -0.001853 \cr
48.94  &  0.101984  &   0.104212  &   -0.002228 \cr
49.70  &  0.111833  &   0.114472  &   -0.002639 \cr
50.46  &  0.121795  &   0.124878  &   -0.003083 \cr
51.22  &  0.131876  &   0.135431  &   -0.003556 \cr
51.96  &  0.142080  &   0.146134  &   -0.004054 \cr
52.70  &  0.152413  &   0.156985  &   -0.004571 \cr
53.43  &  0.162881  &   0.167985  &   -0.005104 \cr
54.16  &  0.173486  &   0.179132  &   -0.005646 \cr
54.88  &  0.184233  &   0.190423  &   -0.006190 \cr
55.60  &  0.195123  &   0.201854  &   -0.006732 \cr
\hline
\end{tabular}
\caption{Results for the 1D jump-diffusion test: $Q_{an}, Q_{FD}$ - the analytical and numerical survival probabilities, $\Delta Q = Q_{anal} - Q_{FD}$.}
\label{TabComp}
\end{center}
\end{table}

\subsection{The two-dimensional problem}
In the first test we solve \eqref{PIDE} with parameters of the model given in Table~\ref{Tab1}. \begin{table}[!ht]
\begin{center}
\begin{tabular}{|c|c|c|c|c|c|c|c|c|c|c|c|c|c|}
\hline
$A_{1,0}$ & $A_{2,0}$ & $L_{1,0}$ & $L_{2,0}$ & $L_{12,0}$ & $L_{21,0}$ & $R_{1}$ & $R_{2}$ & $r$ &  T & $\sigma_1$ & $\sigma_2$ & $\rho$ \cr
\hline
110 & 100 & 80 & 85 & 10 & 15 & 0.4 & 0.35 & 0.05 & 1 & 0.2 & 0.3 & 0.5 \cr
\hline
\end{tabular}
\caption{Parameters of the structural default model.}
\label{Tab1}
\end{center}
\end{table}

For idiosyncratic jumps we chose the Merton model with parameters ($\varphi^i, \mu^i_M, \sigma^i_M), i=1,2$, and for systemic jumps we chose the Kou model with parameters $\varphi_{12}, p, \theta_1, \theta_2$, as shown in Table~\ref{Tab2}. We use the upper script $^{(i)}$ to mark the $i$th bank. Also in these experiments without loss of generality we use $\varphi_1 = \varphi_2 = \varphi_{12} \equiv \varphi$.

We computed all tests using a $100 \times 100$ spatial grid for the convection-diffusion problem. Also we use a constant step in time $\Delta \tau = 0.01$, so that the total number of time steps for a given maturity is also 100. The non-uniform grid for jumps in each direction is a superset of the convection-diffusion grid up to $A_i = 10^{5}$. It is built using a geometric progression and contains 80 nodes.
\begin{table}[!ht]
\begin{center}
\begin{tabular}{|c|c|c|c|c|c|c|c|c|c|c|c|}
\hline
$\varphi$ & $\mu_M^{(1)}$ & $\mu_M^{(2)}$ & $\sigma_M^{(1)}$ & $\sigma_M^{(2)}$ & $p$ & $\theta_1$ & $\theta_2$ & $b_1$ & $b_2$ \cr
\hline
3 & 0.5 & 0.3 & 0.3 & 0.4 & 0.3445 & 3.0465 & 3.0775 & 0.2 & 0.3 \cr
\hline
\end{tabular}
\caption{Parameters of the jump models.}
\label{Tab2}
\end{center}
\end{table}

In Fig.~\ref{Test1} the joint survival probability $Q(x_1,x_2,t,T)$ as computed in our experiment is presented at $t=0$.
\begin{figure}[!ht]
\begin{center}
\fbox{\includegraphics[width=4 in]{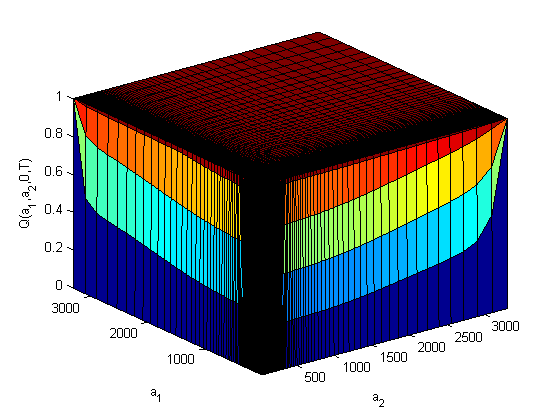}}
\caption{The joint survival probability $Q(x_1,x_2,0,T)$.}
\label{Test1}
\end{center}
\end{figure}

To better see the behavior of the graph close to the initial values of $A_1, A_2$ we zoom-in this picture in the vicinity of these values, see Fig.~\ref{Test1Big}.
\begin{figure}[!ht]
\begin{center}
\fbox{\includegraphics[width=4 in]{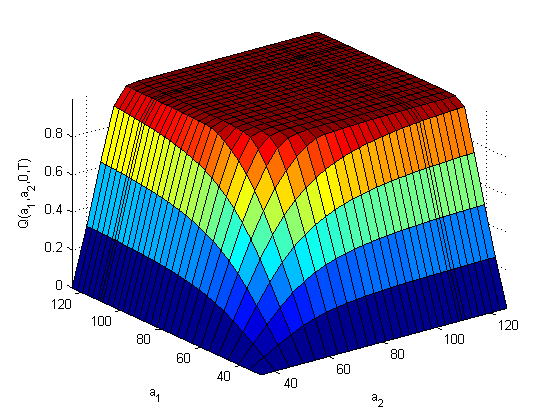}}
\caption{The joint survival probability $Q(x_1,x_2,0,T)$, a zoomed-in picture.}
\label{Test1Big}
\end{center}
\end{figure}

We compare these survival probabilities with those obtained when two banks don't have mutual liabilities. The difference in the corresponding probabilities is shown in Fig.~\ref{TestDif}.
\begin{figure}[!ht]
\begin{center}
\fbox{\includegraphics[width=4 in]{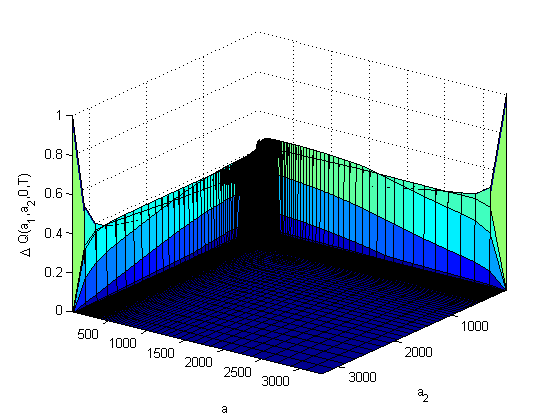}}
\caption{The difference $\Delta Q$ between the joint survival probabilities with and without mutual liabilities.}
\label{TestDif}
\end{center}
\end{figure}

\begin{figure}[!ht]
\begin{center}
\fbox{\includegraphics[width=4.0 in]{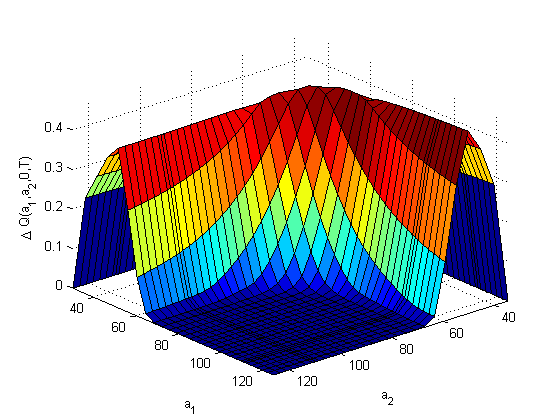}}
\caption{The difference $\Delta Q$, a zoomed-in picture.}
\label{TestDifBig}
\end{center}
\end{figure}

\begin{figure}[!ht]
\begin{center}
\fbox{\includegraphics[width=4.0 in]{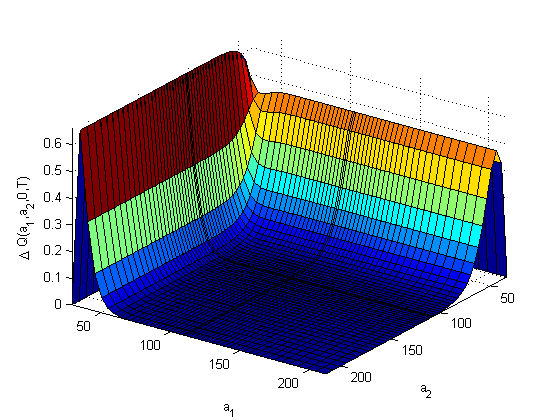}}
\caption{The difference $\Delta Q$ for the pure diffusion case  (the picture is rotated by 180$^\circ$).}
\label{Q2noJumps}
\end{center}
\end{figure}

As expected the maximal difference occurs near default boundaries where the difference could be of order 1. To see how pronounced this effect is, see Fig.~\ref{TestDifBig}. Obviously, the magnitude depends on the values of the jump parameters used in the test as well as on the other parameters of the model and the default boundaries. Also, the effect becomes more pronounced when the ratio of the mutual liabilities to the other liabilities increases.

To emphasize the role of jumps, the same test was conducted without jumps in a pure diffusion setting. The results are shown in Fig.~\ref{Q2noJumps}.  Clearly, the presence of jumps significantly changes the picture, while still preserving the effect of mutual liabilities.

In the second set of tests we setup a local volatility function for assets 1 and 2, which is given in Tables~\ref{locVol1}, \ref{locVol2}
\begin{table}[!ht]
\begin{center}
\ra{1.}
\begin{tabular}{|c|c|c|c|c|c|c|c|c|c|c|c|c|}
\hline
$t, \mbox{yrs}$ & \multicolumn{9}{|c|}{$A_{1,0}$} \\ \cline{2-10}
& 70	& 80 & 90 &	100 & 110 &	120	& 130 &	140 & 150 \\
\hline
0.1 & 0.447 & 0.455 & 0.459 & 0.462 & 0.465 & 0.467 & 0.468 & 0.470 & 0.471\\
\hline
0.2 & 0.500 & 0.507 & 0.511 & 0.514 & 0.516 & 0.518 & 0.519 & 0.520 & 0.522\\
\hline
0.4 & 0.548 & 0.554 & 0.558 & 0.560 & 0.562 & 0.564 & 0.565 & 0.566 & 0.567\\
\hline
0.6 & 0.592 & 0.597 & 0.601 & 0.603 & 0.605 & 0.607 & 0.608 & 0.609 & 0.610\\
\hline
0.8 & 0.632 & 0.638 & 0.641 & 0.643 & 0.645 & 0.646 & 0.648 & 0.649 & 0.650\\
\hline
\end{tabular}
\caption{Local volatility function for $A_{1,t}$.}
\label{locVol1}
\end{center}
\end{table}

\begin{table}[!ht]
\begin{center}
\ra{1.}
\begin{tabular}{|c|c|c|c|c|c|c|c|c|c|c|c|c|}
\hline
$t, \mbox{yrs}$ & \multicolumn{10}{|c|}{$A_{2,0}$} \\ \cline{2-11}
& 50 &	60 & 70	& 80 & 90 &	100 & 110 &	120	& 130 &	140 \\
\hline
0.1 & 0.548 & 0.554 & 0.558 & 0.560 & 0.562 & 0.564 & 0.565 & 0.566 & 0.567 & 0.568 \\
\hline
0.2 & 0.592 & 0.597 & 0.601 & 0.603 & 0.605 & 0.607 & 0.608 & 0.609 & 0.610 & 0.611 \\
\hline
0.4 & 0.632 & 0.638 & 0.641 & 0.643 & 0.645 & 0.646 & 0.648 & 0.649 & 0.650 & 0.650 \\
\hline
0.6 & 0.671 & 0.676 & 0.679 & 0.681 & 0.683 & 0.684 & 0.685 & 0.686 & 0.687 & 0.688 \\
\hline
0.8 & 0.707 & 0.712 & 0.715 & 0.719 & 0.718 & 0.720 & 0.721 & 0.722 & 0.722 & 0.723 \\
\hline
\end{tabular}
\caption{Local volatility function for $A_{2,t}$.}
\label{locVol2}
\end{center}
\end{table}

The results of this test are given in Fig.~\ref{Q2dDif2LocVol}. It can be seen that larger volatilities amplify the effect of mutual liabilities, as well as make a shape of $Q$ highly asymmetric.
\begin{figure}[!ht]
\begin{center}
\fbox{\includegraphics[width=4.0 in]{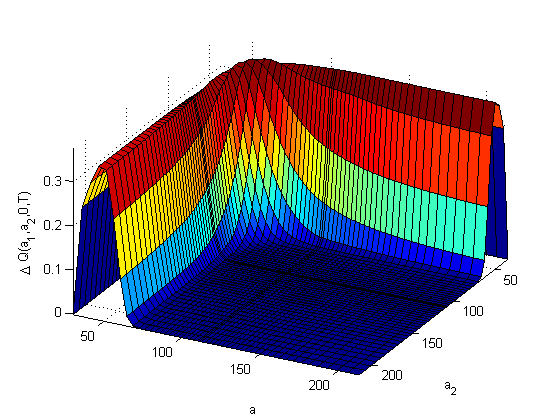}}
\caption{The difference $\Delta Q$ in the presence of local volatility.}
\label{Q2dDif2LocVol}
\end{center}
\end{figure}

We also consider a case of long maturity, $T=10$ years to investigate how the time horizon affects
the shape of the joint survival probability $Q(x_1,x_2,0,T)$
in the presence of mutual liabilities. The corresponding results are shown in Fig.~\ref{Q10}, Fig.~\ref{Rel10}. It is clear that the effect of mutual liabilities significantly decreases when $T$ increases. That is because $Q(x_1,x_2,0,T)$ itself decreases in absolute value with larger $T$, and therefore the absolute value of the effect also drops down.
\begin{figure}[!ht]
\begin{center}
\fbox{\includegraphics[width=4.0 in]{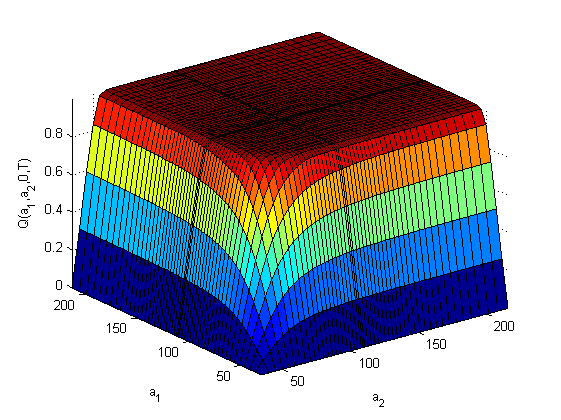}}
\caption{The joint survival probability $Q(a_1,a_2, 0,T)$ at $T=10$ years}
\label{Q10}
\end{center}
\end{figure}

\begin{figure}
\begin{center}
\fbox{\includegraphics[width=4.0 in]{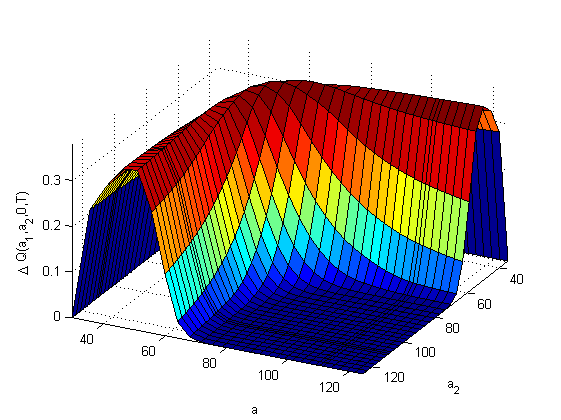}}
\caption{The difference $\Delta Q$ at $T=10$ years with local volatility.}
\label{Rel10}
\end{center}
\end{figure}
The following tests show the influence of correlations on the effects caused by mutual liabilities.   In Fig.~\ref{rho0} the same results as in Test 1 are presented when $\rho=0$, while in Fig.~\ref{loading0} we assume that $b_1 = b_2 = 0$.
\begin{figure}
\begin{center}
\fbox{\includegraphics[width=4.0 in]{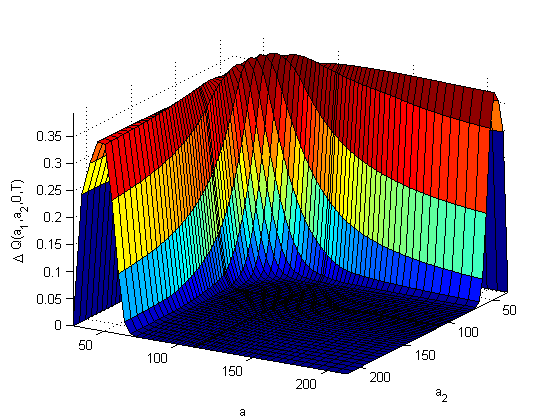}}
\caption{The difference $\Delta Q$ with local volatility and $\rho_{12}=0$.}
\label{rho0}
\end{center}
\end{figure}

\begin{figure}
\begin{center}
\fbox{\includegraphics[width=4.0 in]{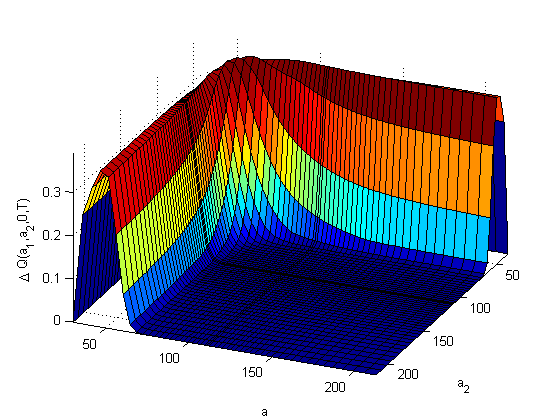}}
\caption{The difference $\Delta Q$ with local volatility and $b_1=b_2=0$.}
\label{loading0}
\end{center}
\end{figure}

These figures show that both contributions of correlations are important.

Fig.~\ref{marg} represents the marginal survival probability of the first bank as a function of the initial asset value of the second bank under the conditions of the first test in Fig.~\ref{Test1}.
And Fig.~\ref{margDif} shows the difference in marginal survival probabilities with and without mutual interbank liabilities.
\begin{figure}
\begin{center}
\fbox{\includegraphics[width=4.0 in]{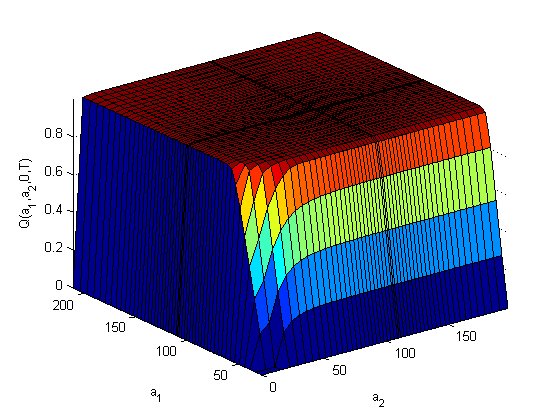}}
\caption{The marginal survival probability $q(x_1,0,T|x_2)$, a zoomed-in picture.}
\label{marg}
\end{center}
\end{figure}

\begin{figure}
\begin{center}
\fbox{\includegraphics[width=4.0 in]{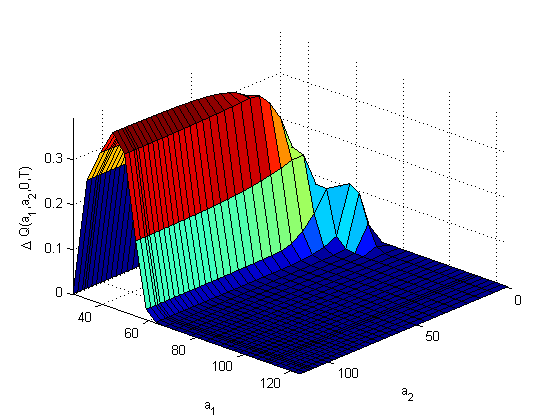}}
\caption{The difference in the marginal survival probabilities $\Delta q(x_1,0,T|x_2)$.}
\label{margDif}
\end{center}
\end{figure}
As could be seen mutual interbank liabilities affect both the marginals and joint survival probabilities. The influence on marginals despite being smaller in magnitude, is still significant.

\section{The three-dimensional case} \label{ne3d}
It is more natural to consider at least three banks, $A_i, \ i=1,2,3$ using the same structural default model as above. Also assume that all three banks have mutual liabilities to each other, as well as liabilities with respect to the outside economy. The advantage of our approach lies in the fact that just minor changes in the computational algorithm need to be done to include the third asset into the whole picture.

Since now $Q = Q(x_1,x_2,x_3, t, T)$, we need to replace the two-dimensional matrices with the three-dimensional ones. Therefore, the expected complexity of the method becomes $O(N_1 N_2 N_3)$. As idiosyncratic jumps are still independent, our splitting algorithm remains the same, although we need to add two more steps in the direction $x_3$ to \eqref{splitAll}. Hence, the 3D splitting algorithm reads
\begin{eqnarray} \label{splitAll3}
Q^{(1)}(x_1,x_2,x_3,\tau) &=& e^{\frac{\Delta \tau}{2} \mathcal{D} } Q(x_1,x_2,x_3,\tau), \\
Q^{(2)}(x_1,x_2,x_3,\tau) &=& e^{\frac{\Delta \tau}{2} \mathcal{J}_1} Q^{(1)}(x_1,x_2,x_3,\tau) \nonumber, \\
Q^{(3)}(x_1,x_2,x_3,\tau) &=& e^{\frac{\Delta \tau}{2} \mathcal{J}_2} Q^{(2)}(x_1,x_2,x_3,\tau) \nonumber, \\
Q^{(4)}(x_1,x_2,x_3,\tau) &=& e^{\frac{\Delta \tau}{2} \mathcal{J}_3} Q^{(3)}(x_1,x_2,x_3,\tau) \nonumber, \\
Q^{(5)}(x_1,x_2,x_3,\tau) &=& e^{\Delta \tau \mathcal{J}_{12}} Q^{(4)}(x_1,x_2,x_3,\tau) \nonumber, \\
Q^{(6)}(x_1,x_2,x_3,\tau) &=& e^{\frac{\Delta \tau}{2} \mathcal{J}_3} Q^{(5)}(x_1,x_2,x_3,\tau) \nonumber, \\
Q^{(7)}(x_1,x_2,x_3,\tau) &=& e^{\frac{\Delta \tau}{2} \mathcal{J}_2} Q^{(6)}(x_1,x_2,x_3,\tau) \nonumber, \\
Q^{(8)}(x_1,x_2,x_3,\tau) &=& e^{\frac{\Delta \tau}{2} \mathcal{J}_1} Q^{(7)}(x_1,x_2,x_3,\tau) \nonumber, \\
Q(x_1,x_2,x_3,\tau + \Delta \tau) &=& e^{\frac{\Delta \tau}{2} \mathcal{D} } Q^{(8)}(x_1,x_2,x_3,\tau).  \nonumber
\end{eqnarray}
In our test experiments at step 5, without loss of generality, we again use the Kou model for the systemic jumps. That requires solving the corresponding 3D linear equations of the first order similar to \eqref{kou1system} and \eqref{kou2system}. The solution could be constructed by using a 3D version of the ADI scheme derived in a similar manner to the 2D case (\cite{McDonough2008}). For the sake of brevity, we formulate two propositions and give just a sketch of the proof since it could be obtained in exactly the same way as in Appendices.

\begin{proposition} \label{PropKou3d_1}
Consider the following PIDE
\begin{equation} \label{PIDEkou3d_1}
(\theta_1- b_1 \triangledown_{x_1} - b_2 \triangledown_{x_2} - b_3 \triangledown_{x_3}) z(x_1,x_2,\tau) = p \theta_1 Q(x_1,x_2,x_3,\tau).
\end{equation}
\noindent and solve it using the following ADI scheme
\begin{align*}
\Big[\Big(s &+ \dfrac{1}{2} \theta_1 \Big) - b_1 \triangledown_{x_1}\Big] z^*({\bf x},\tau) =
\Big[\Big(s - \dfrac{1}{2} \theta_1 \Big) + b_2 \triangledown_{x_2} + b_3 \triangledown_{x_3} \Big]z^k({\bf x},\tau) + b\\
\Big[\Big(s &+ \dfrac{1}{2} \theta_1 \Big) - b_2 \triangledown_{x_2}\Big] z^{**}({\bf x},\tau) =
\Big[\Big(s - \dfrac{1}{2} \theta_1 \Big) + b_1\triangledown_{x_1} + b_3 \triangledown_{x_3}\Big]z^*({\bf x},\tau) + b
\nonumber \\
\Big[\Big(s &+ \dfrac{1}{2} \theta_1 \Big) - b_3 \triangledown_{x_3}\Big] z^{k+1}({\bf x},\tau) =
\Big[\Big(s - \dfrac{1}{2} \theta_1 \Big) + b_1 \triangledown_{x_1} + b_2 \triangledown_{x_2}\Big]z^{**}({\bf x},\tau) + b
\nonumber \\
b &\equiv p \theta_1 Q(x_1,x_2,x_3,\tau) \nonumber
\end{align*}
Then the discrete approximation of this ADI scheme
\begin{align*}
\Big[\Big(s &+ \dfrac{1}{2} \theta_1 \Big)I_{x_1} - b_1 A(x_1)\Big] z^*({\bf x},\tau) =
\Big[\Big(s - \dfrac{1}{2} \theta_1 \Big)I_{x_2} + b_2 A(x_2) + b_3 A(x_3) \Big]z^k({\bf x},\tau) + b\\
\Big[\Big(s &+ \dfrac{1}{2} \theta_1 \Big)I_{x_2} - b_2 A(x_2)\Big] z^{**}({\bf x},\tau) =
\Big[\Big(s - \dfrac{1}{2} \theta_1 \Big)I_{x_1} + b_1 A(x_1) + b_3 A(x_3)\Big]z^*({\bf x},\tau) + b
\nonumber \\
\Big[\Big(s &+ \dfrac{1}{2} \theta_1 \Big)I_{x_3} - b_3 A(x_3)\Big] z^{k+1}({\bf x},\tau) =
\Big[\Big(s - \dfrac{1}{2} \theta_1 \Big)I_{x_1} + b_1 A(x_1) + b_2 A(x_2)\Big]z^{**}({\bf x},\tau) + b
\nonumber \\
b &\equiv p \theta_1 Q(x_1,x_2,x_3,\tau), \qquad A(x_i) =
\begin{cases}
A_2^F(x_i), & b_i > 0 \\
A_2^B(x_i), & b_i < 0, \qquad i=1,2
\end{cases}
\nonumber
\end{align*}
\noindent is unconditionally stable, approximates \eqref{PIDEkou3d_1} with $O\left( \sum_{i,j=1}^3 \Delta x_i \Delta x_j \right)$ and preserves positivity of the solution.
\end{proposition}
\begin{proof}
The proof could be obtained following the lines of Proof of Proposition~\ref{PropKou1} given in Appendix~\ref{ap1}. In our situation we apply the same discretization three times (to each row of the splitting scheme). The remaining part of the proof is exactly same as in Appendix~\ref{ap1}.
\end{proof}
The second proposition is similar in nature.
\begin{proposition} \label{PropKou3d_2}
Consider the following PIDE
\begin{equation} \label{PIDEkou3d_2}
(\theta_2 + b_1 \triangledown_{x_1} + b_2 \triangledown_{x_2} + b_3 \triangledown_{x_3}) z(x_1,x_2,\tau) = (1-p) \theta_2 Q(x_1,x_2,x_3,\tau).
\end{equation}
\noindent and solve it using the following ADI scheme
\begin{align*}
\Big[\Big(s &+ \dfrac{1}{2} \theta_2 \Big) + b_1 \triangledown_{x_1}\Big] z^*(x_1,x_2,\tau) =
\Big[\Big(s - \dfrac{1}{2} \theta_2 \Big) - b_2 \triangledown_{x_2} - b_3 \triangledown_{x_3}\Big]z^k(x_1,x_2,\tau) + b \\
\Big[\Big(s &+ \dfrac{1}{2} \theta_2 \Big) + b_2 \triangledown_{x_2}\Big] z^{**}(x_1,x_2,\tau) =
\Big[\Big(s - \dfrac{1}{2} \theta_2 \Big) - b_1 \triangledown_{x_2} - b_3 \triangledown_{x_3}\Big]z^*(x_1,x_2,\tau) + b
\nonumber \\
\Big[\Big(s &+ \dfrac{1}{2} \theta_2 \Big) + b_3 \triangledown_{x_3}\Big] z^{k+1}(x_1,x_2,\tau) =
\Big[\Big(s - \dfrac{1}{2} \theta_2 \Big) - b_1 \triangledown_{x_1} - b_2 \triangledown_{x_2}\Big]z^{**}(x_1,x_2,\tau) + b
\nonumber \\
b &\equiv (1-p) \theta_2 Q(x_1,x_2,x_3,\tau) \nonumber
\end{align*}

Then the discrete approximation of this ADI scheme
\begin{align*}
\Big[\Big(s &+ \dfrac{1}{2} \theta_2 \Big)I_{x_1} + b_1 A(x_1)\Big] z^*(x_1,x_2,\tau) =
\Big[\Big(s - \dfrac{1}{2} \theta_2 \Big)I_{x_2} - b_2 A(x_2) - b_3 A(x_3)\Big]z^k(x_1,x_2,\tau) + b \\
\Big[\Big(s &+ \dfrac{1}{2} \theta_2 \Big)I_{x_2} + b_2 A(x_2)\Big] z^{**}(x_1,x_2,\tau) =
\Big[\Big(s - \dfrac{1}{2} \theta_2 \Big)I_{x_1} - b_1 A(x_1) - b_3 A(x_3)\Big]z^*(x_1,x_2,\tau) + b
\nonumber \\
\Big[\Big(s &+ \dfrac{1}{2} \theta_2 \Big)I_{x_3} + b_3 A(x_3)\Big] z^{k+1}(x_1,x_2,\tau) =
\Big[\Big(s - \dfrac{1}{2} \theta_2 \Big)I_{x_1} - b_1 A(x_1) - b_2 A(x_2)\Big]z^{**}(x_1,x_2,\tau) + b
\nonumber \\
b &\equiv (1-p) \theta_2 Q(x_1,x_2,x_3,\tau), \qquad
A(x_i)  =
\begin{cases}
A_2^B(x_i), & b_i > 0 \\
A_2^F(x_i), & b_i < 0, \qquad i=1,2
\end{cases}
\nonumber
\end{align*}
\noindent is unconditionally stable, approximates \eqref{PIDEkou3d_2} with $O\left( \sum_{i,j=1}^3 \Delta x_i \Delta x_j \right)$ and preserves positivity of the solution.
\end{proposition}

\begin{proof}
The proof is analogous to that given in Appendix~\ref{ap2} if one applies the same discretization three times (to each row of the splitting scheme). The remaining part of the proof is exactly the same as in Appendix~\ref{ap2}, which in turn is analogous to Appendix~\ref{ap1}.
\end{proof}

The solution of the 3D convection-diffusion problem at the first and the last steps of the scheme is more challenging. So far the unconditional stability of some schemes (Craig-Sneid, Modified Craig-Sneid (MCS),  Hundsdorfer-Verwer (HV), etc.) was proven only when there is no drift term in the corresponding diffusion equation (\cite{HoutMishra2013}). Therefore, this problem requires further attention. Nevertheless, these schemes were successfully used in the 3D setup by \cite{Hout3D} where the MCS and HV schemes demonstrated good stability if the scheme parameter $\theta$ was chosen similar to \cite{HoutMishra2013}.

\subsection{Numerical experiments}
In our tests we chose parameters of the model similar to the 2D case, see Tables \ref{Tab3D1}, \ref{Tab3D2}

\begin{table}[!ht]
\begin{center}
\begin{tabular}{|c|c|c|c|c|c|c|c|c|c|}
\hline
$A_{1,0}$ & $A_{2,0}$ & $A_{3,0}$ & $L_{1,0}$ & $L_{2,0}$ & $L_{3,0}$ & $r$ &  $T$
& $\rho_{xz}$ & $\rho_{yz}$ \cr
\hline
110 & 100 & 120 & 80 & 90 & 100 & 0.05 & 1 & 0.5 & 0.3 \cr
\hline
$\phi_{xy}$ & $L_{12,0}$ & $L_{21,0}$ & $L_{13,0}$ & $L_{31,0}$ & $L_{23,0}$ & $L_{32,0}$ &
$R_{1}$ & $R_{2}$ & $R_{3}$  \cr
\hline
$2\pi/5$ & 20 & 15 & 15 & 20 & 10 & 15 & 0.4 & 0.35 & 0.5 \cr
\hline
\end{tabular}
\caption{Parameters of the 3D structural default model.}
\label{Tab3D1}
\end{center}
\end{table}

\begin{table}[!ht]
\begin{center}
\begin{tabular}{|c|c|c|c|c|c|c|c|c|c|c|c|c|}
\hline
$\varphi$ & $\mu_M^{(1)}$ & $\mu_M^{(2)}$ & $\mu_M^{(3)}$ & $\sigma_M^{(1)}$ & $\sigma_M^{(2)}$ & $\sigma_M^{(3)}$ & $p$ & $\theta_1$ & $\theta_2$ & $b_1$ & $b_2$ & $b_3$ \cr
\hline
3 & 0.5 & 0.3 & 0.4 & 0.3 & 0.4 & 0.5& 0.3445 & 3.0465 & 3.0775 & 0.2 & 0.3 & 0.25 \cr
\hline
\end{tabular}
\caption{Parameters of the 3D jump models.}
\label{Tab3D2}
\end{center}
\end{table}

We recall that a correlation matrix $\Sigma$ of $N$ assets can be represented as a Gram matrix
with matrix elements $\Sigma_{ij} = \left<{\bf x}_i,{\bf x}_j\right>$ where ${\bf x}_i, {\bf x}_j$ are unit vectors on a $N-1$ dimensional hyper-sphere $\mathbb{S}^{N-1}$. Using the 3D geometry, it is easy to establish the following cosine law for the correlations between three assets:
\[ \rho_{xy} = \rho_{yz} \rho_{xz} + \sqrt{(1-\rho_{yz}^2)(1-\rho_{xz}^2)}cos(\phi_{xy}) \]
\noindent with $\phi_{xy}$ being an angle between $\bf x$ and its projection on the plane spanned by ${\bf y}, {\bf z}$. As discussed, e.g., by \cite{Dash2004}, three variables $\rho_{yz}, \rho_{zz}, \phi_{xy}$ are independent, but $\rho_{xy}, \rho_{xz}, \rho_{yz}$ are not. Based on the values given in Tab.~\ref{Tab3D1} we find $\rho_{xy} = 0.4053$.

We compute the test using a $50 \times 50 \times 50$ spatial grid for the convection-diffusion problem. Also we use a constant time step $\Delta \tau = 0.025$, so that the total number of time steps for the given maturity is 40. The jump non-uniform grid in each direction is a superset of the convection-diffusion grid up to $A_i = 10^{4}$ built using a geometric progression. So the jumps are computed on the grid with $62 \times 64 \times 63$ nodes. Also we chose $s_i = \theta_i + 1, \ i \in [1,3]$ which provided convergence of the ADI scheme for the common jumps after 4 iterations.

We again compare the survival probability in the presence of mutual liabilities,  $Q^w(A_1, A_2$, $A_3)$, with that in the absence of mutual liabilities, $Q^{wo}(A_1,A_2,A_3)$. To obtain the latter, we first reduce $L_{1,0}, L_{2,0}, L_{3,0}$ by the amounts $L_{ij,0}, \ i \in [1,3], \  j \in [1,3], i \ne j$, and then put $L_{ij,t} = 0$. The difference $\Delta Q = Q^w - Q^{wo}$ is presented in Fig.~\ref{Fig3DXY} - \ref{Fig3DYZ}. Since the whole picture in this case is four-dimensional, we represent it as a series of 3D projections, namely: Fig.~\ref{Fig3DXY} represents the $A_1 - A_2$ plane at various values of the $A_3$ coordinate which are indicated in the corresponding labels; Fig.~\ref{Fig3DXZ} does same in the $A_1 - A_3$ plane, and Fig.~\ref{Fig3DYZ} - in the $A_2 - A_3$ plane.

\begin{figure}[!ht]
\begin{center}
\vspace{-1.in}
\begin{turn}{90}
\includegraphics[width=9.5 in]{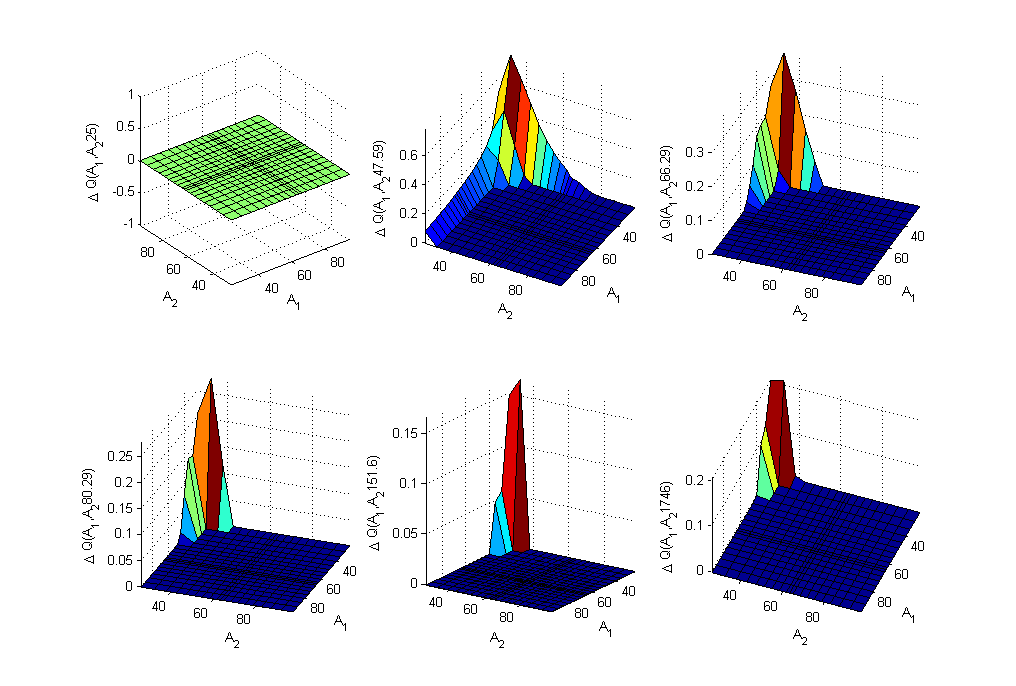}
\end{turn}
\caption{The difference $\Delta Q$ with and without mutual liabilities, $A_1 - A_2$ plane.}
\label{Fig3DXY}
\end{center}
\end{figure}

\begin{figure}[!ht]
\begin{center}
\vspace{-1.in}
\begin{turn}{90}
\includegraphics[width=9.5 in]{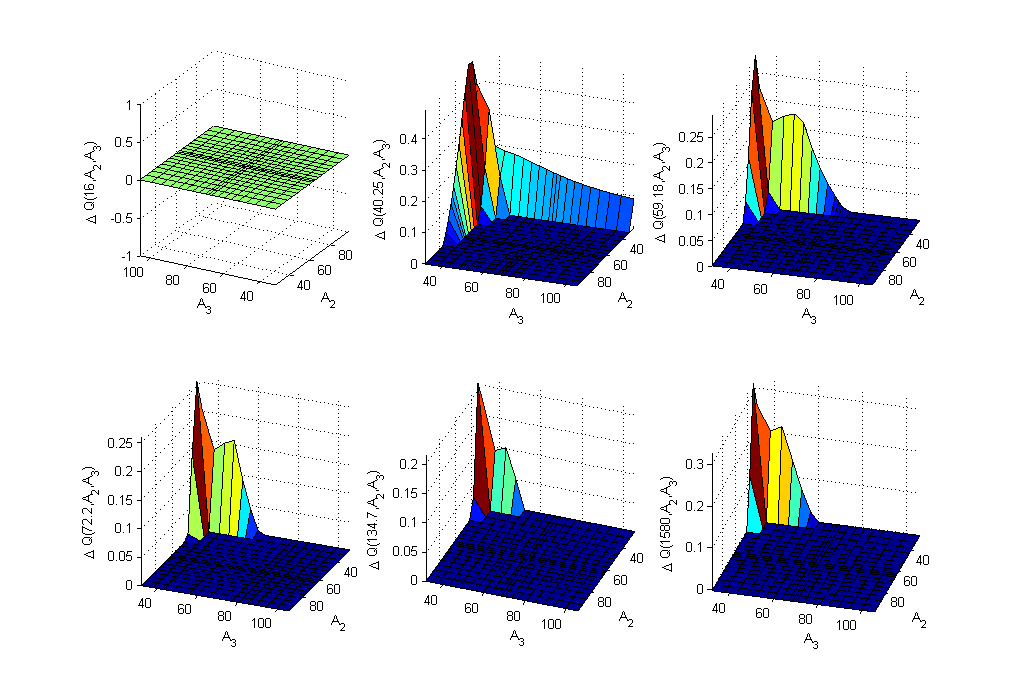}
\end{turn}
\caption{The difference $\Delta Q$ with and without mutual liabilities, $A_1 - A_3$ plane.}
\label{Fig3DXZ}
\end{center}
\end{figure}

\begin{figure}[!ht]
\begin{center}
\vspace{-1.in}
\begin{turn}{90}
\includegraphics[width=9.5 in]{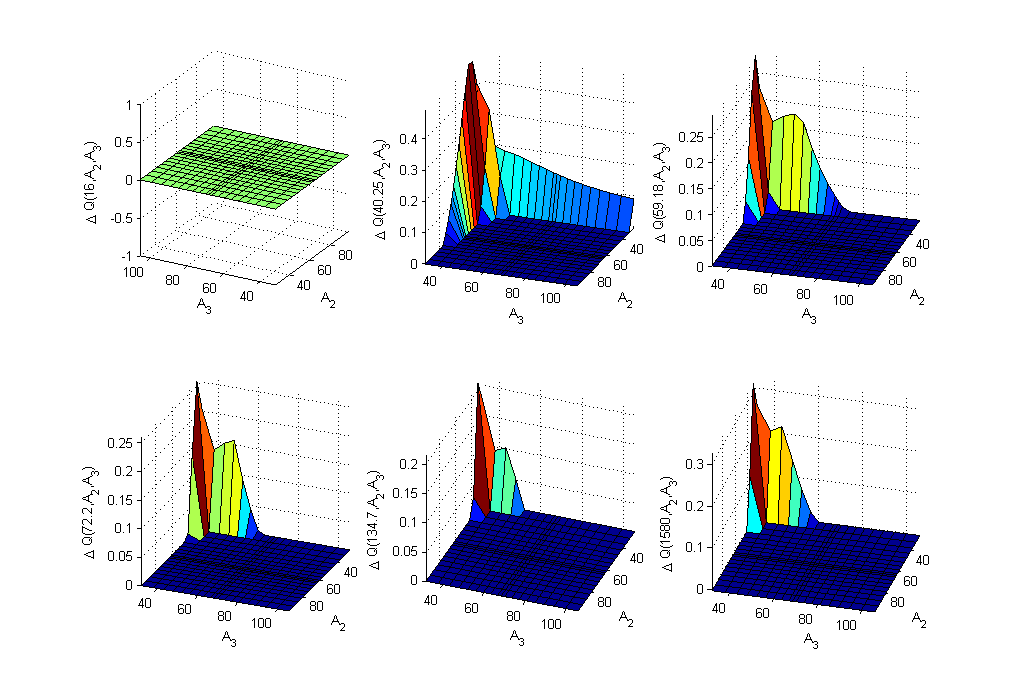}
\end{turn}
\caption{The difference $\Delta Q$ with and without mutual liabilities, $A_2 - A_3$ plane.}
\label{Fig3DYZ}
\end{center}
\end{figure}

Two observations could be made based on the results obtained in these tests. First, when three banks have mutual liabilities, their effect on the joint survival probability is more profound than in the 2D case. Second, $\Delta Q$ has an irregular shape as a function of 3 coordinates. For instance, in the $A_2 - A_3$ plane it has two local maxima (in the absolute value) while in the 2D case it doesn't demonstrate such a behavior. Also this effect disappears in the absence of jumps. This is similar to the effect observed in \cite{ItkinCarrBarrierR3} where asymmetric positive and negative jumps in the stochastic skew model were described by the CGMY model with different $\alpha$, which produced a qualitatively new effect. It is evident through the appearance of a big dome close to the ATM at the moderate values of the instantaneous variance $v$ in addition to a standard arc of the double barrier options which is also close to the ATM, but at small values of $v$.

As expected, the whole picture is rather complicated. Moreover, as it is affected by the number of model parameters, which could be difficult to extract from a set of liquid market data, it could be
very challenging to calibrate such a model. A standard recipe is to first calibrate marginals of the distribution to the corresponding market data, and then use some other data for calibration of the remaining parameters.

\section{Conclusions} \label{conclusion}

In this paper we presented three main innovations which seem to be rather general, namely:
\begin{enumerate}
\item We introduced mutual banks' liabilities into the structural default model. We discussed how these liabilities affect joint and marginal survival probabilities, and provided some numerical test results. These results demonstrate that the effect of mutual liabilities could be quite significant. Of course, the magnitude of the effect depends on how close the initial asset values are to the default barrier, and parameters which describe the assets' dynamics, such  as volatility, etc. These parameters, in principle, could be found by calibrating marginal survival probabilities to market CDS spreads.

\item To make the above analysis tractable we developed a solution scheme for the model considering a set of banks with mutual interbank liabilities whose assets are driven by correlated L\'evy processes.
For every asset, the jumps are represented as a weighted sum of the common and idiosyncratic parts. Both parts could be simulated by an arbitrary L\'evy model which is an extension of the previous approaches where either the discrete or exponential jumps were considered, or a L\'evy copula approach was utilized. We provided a novel efficient (linear complexity in each dimension) numerical (splitting) algorithm for solving the corresponding 2D and 3D jump-diffusion equations, and proved its convergence and second order of accuracy in both space and time.

\item The joint survival probability of three firms $Q(x_1,x_2,x_3,0,T)$ was computed using the above framework. To the best of our knowledge there were no the similar results reported in literature. We found that in some cases, the difference between the joint survival probabilities with and without mutual liabilities has a bimodal profile in some projections, and this effect disappears in the pure diffusion setup. This is similar to what was observed in \cite{ItkinCarrBarrierR3} where interaction of jumps also produced a bimodal distribution for double barrier option prices.
\end{enumerate}

Despite the fact that the present approach is efficient and attractive in low dimensions, it is not clear how best to extend it to the case when the number of firms is more than three, unless some simplifications are introduced into the model. This is a standard limitation of the FD approach which experiences the curse of dimensionality. A possible way to overcome this could be to combine the analytical and numerical methods, similar to how this was done in, e.g., \cite{LiptonSavescu2014}.

\clearpage
\section*{Acknowledgments}
We thank Peter Carr, Darrel Duffie, Peter Forsyth, Igor Halperin and Rajeev Virmani for useful comments. We assume full responsibility for any remaining errors.


\newcommand{\noopsort}[1]{} \newcommand{\printfirst}[2]{#1}
  \newcommand{\singleletter}[1]{#1} \newcommand{\switchargs}[2]{#2#1}

\clearpage
\appendix

\section{Proof of Proposition~\protect{\ref{PropKou1}}} \label{ap1}
Following \cite{ElhashashSzyld2008}, we introduce definition of an EM-matrix
\begin{definition}
An $N\times N$ matrix $A = [a_{ij}]$ is called an EM-Matrix if it can be represented as $A = sI - B$ with $0 < \rho(B) < s$, $s > 0$ is some constant, $\rho(B)$ is the spectral radius of $B$, and $B$ is an eventually nonnegative matrix.
\end{definition}

Now suppose $b_1 < 0, \ b_2 < 0$. Then the matrix
\[ M_1 = \Big[\Big(s + \dfrac{1}{2} \theta_1 \Big)I_{x_1} - b_1 A^B_2(x_1)\Big] \]
in the first row of \eqref{adi1d} is an EM-matrix, see Lemma A.2 in \cite{Itkin2014a}. Therefore, the inverse of $M_1$ is a non-negative matrix, see Lemma A.3 in \cite{Itkin2014a}.

The matrix
\[ M_2 = \Big[\Big(s - \dfrac{1}{2} \theta_1 \Big)I_{x_2} + b_2 A_2^B(x_2)\Big] \]
is an eventually non-negative matrix\footnote{By definition of $A_2^B$ the matrix $M_2$ is a lower triangular matrix with three non-zero diagonals. The main and the first lower diagonals are positive and the second lower diagonal is negative. However, the former two dominate the latter one.} if $s$ is chosen to provide
\begin{equation} \label{conv1}
s > \dfrac{1}{2} \theta_1 - b_2 \dfrac{3}{h}
\end{equation}
Therefore, the solution of the first row of \eqref{adi1d} is $z^*(x_1,x_2,\tau) = M_1^{-1}\left[M_2 z^k(x_1,x_2, \tau) + b\right]$ which by construction is a non-negative vector.
Also eigenvalues of $M_1^{-1}M_2$ are
\[ \lambda_i = \dfrac{s - \dfrac{1}{2} \theta_1 + 3 b_2/h }{s + \dfrac{1}{2} \theta_1 - 3 b_1/h } < 1, \qquad i \in [1,N_1] \]
Therefore, this scheme converges unconditionally provided \eqref{conv1} is satisfied.

Also, by construction  the matrix $A_2^B(x)$ approximates the operator $\triangledown_x$ to the second order, i.e., with $O(h^2)$. Therefore, the whole scheme provides the second order approximation.

The second row of \eqref{adi1d} could be analyzed in the same way.

In all other cases $b_1 < 0, b_2 > 0$, $b_1 > 0, b_2 < 0$ and $b_1 > 0, b_2 > 0$ the proof could be done by analogy.

\section{Proof of Proposition~\protect{\ref{PropKou2}}} \label{ap2}
The proof is completely analogous to that in Appendix~\ref{ap1}.

\section{Matrix exponential approach for exponential jumps} \label{ap4}
In the 1D case we still want to use the splitting algorithm of \eqref{splitFin}. To proceed, let us define an explicit model for jumps, so the pseudo-differential operator $\mathcal{J}$ defined in \eqref{intGen} could be computed explicitly.

Let us consider only negative exponentially distributed jumps\footnote{For the positive jumps this could be done in a similar way. The denominator in \eqref{expScheme} then changes to $\phi-1$ and the term $\phi I + a A^B_2$ changes to $\phi I - a A^F_2$ where $\phi > 1$.}, see \cite{Lipton2002a}, i.e.
\begin{equation} \label{expJ}
\nu(J) =
\begin{cases}
\phi e^{\phi J}, & J \le 0 \cr
0, & J > 0,
\end{cases}
\end{equation}
\noindent where $\phi > 0$ is the parameter of the exponential distribution. With the L\'evy measure $\nu(dy)$ given in \eqref{expJ} and the intensity of jumps $\lambda \ge 0$  we can substitute $\nu(dy)$ into \eqref{intGen} and integrate. The result reads
\begin{equation} \label{expJI}
\mathcal{J} = \dfrac{\lambda}{\phi+1}(\phi + \triangledown_x)^{-1} (\triangledown_x^2 - \triangledown_x), \qquad \triangledown_x \equiv \partial_x.
\end{equation}
Below for simplicity of notation we introduce $a \equiv A_1$. Since $x = \log a$, the above expression could be re-written as
\begin{equation} \label{expJI1}
\mathcal{J} = \dfrac{\lambda}{\phi+1}(\phi + a \triangledown)^{-1} a^2\triangledown^2, \qquad \triangledown \equiv \partial_a.
\end{equation}

\begin{proposition} \label{PropExp}
Consider the following discrete approximation of \eqref{expJI1}:
\begin{align} \label{expScheme}
J = \dfrac{\lambda}{\phi+1}(\phi I + a A^B_2)^{-1} a^2 A^C_2.
\end{align}
Then this scheme is a) unconditionally stable; b) approximates the operator $\mathcal{J}$ in \eqref{expJI1} on a certain non-uniform grid in variable $a$ with $O(\max(h_i)^2)$, where $h_i, \ i=1,...,N$ are the steps of the grid; c) and preserves positivity of the solution.
\end{proposition}
\begin{proof}
For the sake of clarity we give the proof for the uniform grid, as an extension to the non-uniform grid is straightforward.

As shown in \cite{Itkin2014a} the matrix $A^B_2$ is an EM matrix. Therefore, the matrix $\phi I + a A^B_2$ is also an EM-matrix. Therefore, its inverse is a non-negative matrix. The matrix $A^C_2$ by construction is the Metzler matrix. A product of the non-negative and Metzler matrices is the negative of an EM-matrix\footnote{Some care should be taken regarding the boundary values of $A^C_2$ to guarantee this. Usually, introduction of ghost points at the boundaries helps to increase the accuracy of the method. Alternatively, one could use another approximation of the term $(\phi + a \triangledown)^{-1}$ in \eqref{expJI1} which is $(\phi I + a A^B)^{-1}$. This reduces the order of approximation from the exact second order to some order in between 1 and 2, but, at the same time, significantly improves the properties of the resulting matrix $J$.}. As $\phi > 0$, the matrix $J$ is also the negative of an EM-matrix. Then unconditional stability and positivity of the solution follows from the main Theorem in \cite{Itkin2014a}. As the matrix $A^F_2$ is the second order approximation in $h$ to $\triangledown$, and $ A^C_2$ is the second order approximation in $h$ to $\triangledown^2$, the whole scheme approximates the operator $\mathcal{J}$ with the second order in $h$.
\end{proof}

In practical applications the complexity of this scheme could be linear in the number of grid nodes $N$. Indeed, suppose we wish to compute $Q$ with the second order of approximation in the time step $\Delta t$, i.e. with the accuracy $O((\Delta t)^2)$. Represent $e^{\Delta t \lambda J}$ in the second step of the splitting algorithm \eqref{splitFin} using a Pad\'e rational approximation $(1,1)$:
\[ e^{\Delta t \lambda J} = \left(I - \dfrac{1}{2}\lambda \Delta t J\right)^{-1}\left(I + \dfrac{1}{2}\lambda \Delta t J\right) \]
With allowance for \eqref{expScheme} after some algebra this could be re-written in the form
\begin{equation} \label{Pade}
\left[\phi I - a A^F_2 - \dfrac{\lambda}{2(\phi-1)} \Delta t a^2 A^C_2\right] Q(a,t+\Delta t) =
\left[\phi I - a A^F_2 + \dfrac{\lambda}{2(\phi-1)} \Delta t a^2 A^C_2\right] Q(a,t).
\end{equation}
Matrices in square brackets are banded (three or five diagonal), therefore this system of linear equations could be solved with the complexity $O(N)$.

\end{document}